\documentclass[11pt]{article}
\usepackage{hyperref}
\usepackage{times}  
\usepackage{mathpazo}
\usepackage{amssymb,amsmath,amsthm}
\usepackage{epsfig}
\usepackage{tcolorbox}
\usepackage{lipsum} 
\usepackage{enumitem}

\newtheorem{theorem}{Theorem}[section]
\newtheorem{proposition}[theorem]{Proposition}
\newtheorem{definition}[theorem]{Definition}

\newtheorem{claim}[theorem]{Claim}
\newtheorem{lemma}[theorem]{Lemma}

\newtheorem{corollary}[theorem]{Corollary}

\newtheorem{remark}[theorem]{Remark}

\newcommand{\qedsymb}{\hfill{\rule{2mm}{2mm}}}
\renewenvironment{proof}[1][]{\begin{trivlist}
\item[\hspace{\labelsep}{\bf\noindent Proof#1:\/}] }{\qedsymb\end{trivlist}}

\def\calG{{\cal G}}

\def\calW{{\cal W}}

\def\R{\mathbb{R}}

\def\N{\mathbb{N}}

\newcommand{\rectbinom}[2]{\genfrac{[}{]}{0pt}{}{#1}{#2}}

\newcommand{\PSPACE}{\mathsf{PSPACE}}
\newcommand{\NP}{\mathsf{NP}}

\renewcommand{\P}{\mathsf{P}}

\newcommand{\coNPpoly}{\mathsf{coNP/poly}}

\newcommand{\YES}{\textsc{YES}}
\newcommand{\NO}{\textsc{NO}}

\newcommand{\od}{\overline{\xi}}
\newcommand{\KG}{\mathsf{K}}

\newcommand{\eps}{\epsilon}
\renewcommand{\epsilon}{\varepsilon}

\newcommand{\linspan}{\mathop{\mathrm{span}}}
\newcommand{\Fset}{\mathbb{F}}         

\newcommand{\SAT}{\textsc{SAT}}

\newcommand{\Col}{\textsc{Coloring}}
\newcommand{\qCol}{\textsc{-Coloring}}

\newcommand{\qListCol}{\textsc{-List Coloring}}
\newcommand{\SubChoose}{\textsc{Subspace Choosability}}
\newcommand{\dSubChoose}{\textsc{-Subspace Choosability}}
\newcommand{\ODP}{\textsc{Ortho-Dim}}
\newcommand{\dODP}{\textsc{-Ortho-Dim}}
\newcommand{\Empty}{\textsc{Empty}}
\newcommand{\Path}{\textsc{Path}}
\newcommand{\Split}{\textsc{Split}}
\newcommand{\USplit}{\cup \textsc{Split}}
\newcommand{\Cochordal}{\textsc{Cochordal}}
\newcommand{\UCochordal}{\cup \textsc{Cochordal}}

\begin{document}

\title{{\bf Kernelization for Orthogonality Dimension}}

\author{
Ishay Haviv\thanks{School of Computer Science, The Academic College of Tel Aviv-Yaffo, Tel Aviv 61083, Israel. Research supported by the Israel Science Foundation (grant No.~1218/20).}
\and
Dror Rabinovich\footnotemark[1]
}

\date{}

\maketitle

\begin{abstract}
The orthogonality dimension of a graph over $\R$ is the smallest integer $d$ for which one can assign to every vertex a nonzero vector in $\R^d$ such that every two adjacent vertices receive orthogonal vectors. For an integer $d$, the $d\dODP_\R$ problem asks to decide whether the orthogonality dimension of a given graph over $\R$ is at most $d$.
We prove that for every integer $d \geq 3$, the $d\dODP_\R$ problem parameterized by the vertex cover number $k$ admits a kernel with $O(k^{d-1})$ vertices and bit-size $O(k^{d-1} \cdot \log k)$. We complement this result by a nearly matching lower bound, showing that for any $\eps > 0$, the problem admits no kernel of bit-size $O(k^{d-1-\eps})$ unless $\NP \subseteq \coNPpoly$. We further study the kernelizability of orthogonality dimension problems in additional settings, including over general fields and under various structural parameterizations.
\end{abstract}

\section{Introduction}

For a field $\Fset$ and an integer $d$, a $d$-dimensional orthogonal representation of a graph $G=(V,E)$ over $\Fset$ is an assignment of a vector $u_v \in \mathbb{F}^d$ with $\langle u_v,u_v \rangle \neq 0$ to each vertex $v \in V$, such that for every two adjacent vertices $v$ and $v'$ in $G$, it holds that $\langle u_{v} , u_{v'} \rangle = 0$.
We consider here the standard inner product, defined for any two vectors $x,y \in \mathbb{F}^d$ by $\langle x , y \rangle = \sum_{i=1}^{d}{x_i \cdot y_i}$ with operations performed over the field $\Fset$.
The orthogonality dimension of a graph $G$ over $\Fset$, denoted by $\od_\Fset(G)$, is the smallest integer $d$ for which $G$ admits a $d$-dimensional orthogonal representation over $\Fset$ (see Definition~\ref{def:OD} and Remark~\ref{remark:OD}).

The notion of orthogonal representations over the real field $\R$ was introduced in 1979 by Lov\'asz~\cite{Lovasz79}, who used them to define the celebrated $\vartheta$-function that was motivated by questions in information theory on the Shannon capacity of graphs.
Over the years, orthogonal representations and the orthogonality dimension have found a variety of applications in several areas of research.
In graph theory, orthogonal representations over the reals were used by Lov\'asz, Saks, and Schrijver~\cite{LovaszSS89} to characterize connectivity properties of graphs (see also~\cite[Chapter~10]{LovaszBook}).
In computational complexity, the orthogonality dimension over finite fields was related to lower bounds in circuit complexity by Codenotti, Pudl{\'{a}}k, and Resta~\cite{CodenottiPR00} (see also~\cite{GolovnevRW17,GolovnevH20}). Over the complex field, it was employed by de Wolf~\cite{deWolfThesis} to determine the quantum one-round communication complexity of promise equality problems (see also~\cite{CameronMNSW07,BrietBLPS15,BrietZ17}).
Additional notable applications from the area of information theory are related to index coding~\cite{BirkKol98,BBJK06}, distributed storage~\cite{BargZ22}, and hat-guessing games~\cite{Riis07}.

The question of the complexity of determining the orthogonality dimension of a given graph over a specified field was proposed in 1989 by Lov\'asz et al.~\cite{LovaszSS89}.
For a field $\Fset$ and an integer $d$, consider the decision problem that given a graph $G$ asks to decide whether $\od_\Fset(G) \leq d$.
It is easy to see that the problem can be solved efficiently for $d \in \{1,2\}$, because a graph $G$ satisfies $\od_\Fset(G) \leq 1$ if and only if it is edgeless, and it satisfies $\od_\Fset(G) \leq 2$ if and only if it is bipartite. For every $d \geq 3$, however, it was shown by Peeters~\cite{Peeters96} in 1996 that the problem is $\NP$-hard for every field $\Fset$. More recently, it was shown in~\cite{ChawinH23} that for every sufficiently large integer $d$, it is $\NP$-hard to distinguish graphs $G$ that satisfy $\od_\Fset(G) \leq d$ from those satisfying $\od_\Fset(G) \geq 2^{(1-o(1)) \cdot d/2}$, provided that $\Fset$ is either a finite field or $\R$.

Motivated by the diverse applications of orthogonality dimension, the present paper delves into the computational complexity of this graph quantity from the perspective of parameterized complexity.
We study the decision problems associated with orthogonality dimension with respect to various structural parameterizations, with a particular attention dedicated to the vertex cover number parameterization. We exhibit fixed-parameter tractability results for such problems, along with upper and lower bounds on their kernelizability.
Our approach draws its inspiration from prior work on the parameterized complexity of coloring problems, most notably by Jansen and Kratsch~\cite{JansenK13} (see also~\cite[Chapter~7]{JansenThesis}) and by Jansen and Pieterse~\cite{JansenP19color}.
In what follows, we provide an overview on relevant research on coloring problems and then turn to a description of our contribution.
We use here standard notions from the area of parameterized complexity, whose definitions can be found in Section~\ref{sec:parameter} (see also, e.g.,~\cite{CyganFKLMPPS15,KernelBook19}).

Graph coloring is a cornerstone concept in graph theory that has been extensively studied from a computational point of view.
For an integer $q$, a $q$-coloring of a graph $G$ is an assignment of a color to each vertex of $G$ from a set of $q$ colors.
The coloring is said to be proper if it assigns distinct colors to every two adjacent vertices in the graph. A graph $G$ is called $q$-colorable if it admits a proper $q$-coloring, and the smallest integer $q$ for which $G$ is $q$-colorable is called the chromatic number of $G$ and is denoted by $\chi(G)$.
For an integer $q$, let $q\qCol$ denote the decision problem that given a graph $G$ asks to decide whether $\chi(G) \leq q$.
The $\Col$ problem is defined similarly, with the only, yet crucial, difference that the number of colors $q$ is not fixed but forms part of the input.
It is well known that the $q\qCol$ problem can be solved in polynomial time for $q \in \{1,2\}$ and is $\NP$-complete for every $q \geq 3$.
This implies that the $\Col$ problem, parameterized by the number of colors $q$, is not fixed-parameter tractable unless $\P = \NP$.

The study of the parameterized complexity of coloring problems was initiated in 2003 by Cai~\cite{Cai03}, who proposed the following terminology.
For a family of graphs $\calG$ and for an integer $k$, let $\calG+k\mathrm{v}$ denote the family of all graphs that can be obtained from a graph of $\calG$ by adding at most $k$ vertices (with arbitrary neighborhoods). Equivalently, a graph $G=(V,E)$ lies in $\calG+k\mathrm{v}$ if there exists a set $X \subseteq V$ of size $|X| \leq k$, referred to as a modulator, such that the graph $G \setminus X$ obtained from $G$ by removing the vertices of $X$ lies in $\calG$.
For example, letting $\Empty$ denote the family of all edgeless graphs, the family $\Empty+k\mathrm{v}$ consists of all the graphs that admit a vertex cover of size at most $k$.
For an integer $q$, the $q\qCol$ problem on $\calG+k\mathrm{v}$ graphs is the parameterized problem defined as follows.
\begin{enumerate}
  \setlength{\itemsep}{0em} 
  \setlength{\parsep}{0em}  
  \item[] Input: A graph $G=(V,E)$ and a set $X \subseteq V$ such that $G \setminus X \in \calG$.
  \item[] Question: Is $\chi(G) \leq q$?
  \item[] Parameter: The size $k=|X|$ of the modulator $X$.
\end{enumerate}
As before, the $\Col$ problem on $\calG+k\mathrm{v}$ graphs is defined similarly, except that the number of colors $q$ is not fixed but forms part of the input.

A common parameterization of graph problems that received a considerable amount of attention in the literature is that of the vertex cover number, corresponding to the family $\calG = \Empty$ (see, e.g.,~\cite{FominJP14}).
It is well known that the $\Col$ problem on $\Empty+k\mathrm{v}$ graphs is fixed-parameter tractable.
Nevertheless, Bodlaender, Jansen, and Kratsch~\cite{BodlaenderJK14} proved that the problem does not admit a kernel of polynomial size under the assumption $\NP \nsubseteq \coNPpoly$, whose refutation is known to imply the collapse of the polynomial-time hierarchy~\cite{Yap83}.
Yet, for any fixed integer $q \geq 3$, Jansen and Kratsch~\cite{JansenK13} showed that the $q\qCol$ problem on $\Empty+k\mathrm{v}$ graphs admits a kernel with $O(k^q)$ vertices which can be encoded in $O(k^q)$ bits (see also~\cite{FominJP14}).
This result was improved by Jansen and Pieterse~\cite{JansenP19color} as an application of an algebraic sparsification technique they introduced in~\cite{JansenP19sparse}. It was shown in~\cite{JansenP19color} that for every $q \geq 3$, the $q\qCol$ problem on $\Empty+k\mathrm{v}$ graphs admits a kernel with $O(k^{q-1})$ vertices and bit-size $O(k^{q-1} \cdot \log k)$ (see~\cite{JansenP19color} for various generalizations).
On the contrary, it was shown in~\cite{JansenK13,JansenP19color} that for every $q \geq 3$ and any $\eps > 0$, the problem does not admit a kernel that can be encoded in $O(k^{q-1-\eps})$ bits unless $\NP \subseteq \coNPpoly$, thereby settling its kernelization complexity up to a multiplicative $k^{o(1)}$ term.

The paper~\cite{JansenK13} further studied the kernelization complexity of the $q\qCol$ problem on $\calG+k\mathrm{v}$ graphs for general families $\calG$. In particular, they considered graph families $\calG$ that are hereditary (i.e., closed under removal of vertices) and that are, roughly speaking, local with respect to the $q\qListCol$ problem, in the sense that a $\NO$ instance of $q\qListCol$ that involves a graph from $\calG$ must have a $\NO$ sub-instance whose size depends solely on $q$. For such families $\calG$, it was shown in~\cite{JansenK13} that the $q\qCol$ problem on $\calG+k\mathrm{v}$ graphs admits a kernel of polynomial size, and this result was complemented with a lower bound on the kernel size relying on the assumption $\NP \nsubseteq \coNPpoly$.
These results apply, for example, for the families $\USplit$ and $\UCochordal$ of the graphs whose connected components are split graphs and cochordal graphs respectively (see Definitions~\ref{def:split} and~\ref{def:cochordal}).
On the other hand, strengthening a result of Bodlaender et al.~\cite{BodlaenderDFH09}, the authors of~\cite{JansenK13} proved that the $3\qCol$ problem on $\Path+k\mathrm{v}$ graphs does not admit a kernel of polynomial size unless $\NP \subseteq \coNPpoly$, where $\Path$ stands for the family of path graphs.

\subsection{Our Contribution}

This paper initiates a systematic study of the parameterized complexity of the orthogonality dimension of graphs.
It is noteworthy that the orthogonality dimension of graphs is related to their chromatic number.
Indeed, for every field $\Fset$ and for every graph $G$, it holds that $\od_\Fset(G) \leq \chi(G)$, because a proper $q$-coloring of $G$ may be viewed as a $q$-dimensional orthogonal representation of $G$ over $\Fset$ that uses only vectors from the standard basis of $\Fset^q$ (see Claim~\ref{claim:OD<=CHI}).
Yet, it turns out that the two graph quantities can differ substantially, as there exist graphs where the orthogonality dimension is exponentially smaller than the chromatic number (see, e.g.,~\cite[Proposition~2.2]{HavivMFCS19}).
Our investigation of the parameterized complexity of orthogonality dimension problems aligns with the approach of~\cite{JansenK13,JansenP19color} for studying coloring problems within the parameterized complexity framework. While coloring problems are primarily combinatorial in nature, the attempt to prove analogous results for orthogonality dimension raises intriguing questions reflecting the algebraic aspects of this graph quantity.

We first introduce the decision problems associated with orthogonality dimension.
\begin{definition}\label{def:OD_P}
For a field $\Fset$, the $\ODP_\Fset$ problem is defined as follows.
\begin{enumerate}
  \setlength{\itemsep}{0em} 
  \setlength{\parsep}{0em}  
  \item[] Input: A graph $G=(V,E)$ and an integer $d$.
  \item[] Question: Is $\od_\Fset(G) \leq d$?
\end{enumerate}
For a field $\Fset$ and a family of graphs $\calG$, the (parameterized) $\ODP_\Fset$ problem on $\calG+k\mathrm{v}$ graphs is defined as follows.
\begin{enumerate}
  \setlength{\itemsep}{0em} 
  \setlength{\parsep}{0em}  
  \item[] Input: A graph $G=(V,E)$, a set $X \subseteq V$ such that $G \setminus X \in \calG$, and an integer $d$.
  \item[] Question: Is $\od_\Fset(G) \leq d$?
  \item[] Parameter: The size $k=|X|$ of the modulator $X$.
\end{enumerate}
For a field $\Fset$, an integer $d$, and a family of graphs $\calG$, the (parameterized) $d\dODP_\Fset$ problem on $\calG+k\mathrm{v}$ graphs is defined as follows.
\begin{enumerate}
  \setlength{\itemsep}{0em} 
  \setlength{\parsep}{0em}  
  \item[] Input: A graph $G=(V,E)$ and a set $X \subseteq V$ such that $G \setminus X \in \calG$.
  \item[] Question: Is $\od_\Fset(G) \leq d$?
  \item[] Parameter: The size $k=|X|$ of the modulator $X$.
\end{enumerate}
\end{definition}
\noindent
Let us stress that the integer $d$ forms part of the input in the $\ODP_\Fset$ problem, whereas it is a fixed constant in the $d\dODP_\Fset$ problem.
Note that the hardness result of~\cite{Peeters96} implies that for every field $\Fset$, the $\ODP_\Fset$ problem parameterized by the solution value $d$ is not fixed-parameter tractable unless $\P=\NP$.

The main parameterization we consider is the vertex cover number of the input graph, which corresponds to the family $\calG = \Empty$.
We start with the following fixed-parameter tractability result.

\begin{theorem}\label{thm:IntroFPT}
Let $\Fset$ be either a finite field or $\R$.
The $\ODP_\Fset$ problem on $\Empty+k\mathrm{v}$ graphs is fixed-parameter tractable.
\end{theorem}
\noindent
In fact, we prove an extension of Theorem~\ref{thm:IntroFPT}, showing that if the $\ODP_\Fset$ problem over a field $\Fset$ is decidable, then the corresponding $\ODP_\Fset$ problem on $\Empty+k\mathrm{v}$ graphs is fixed-parameter tractable (see Theorem~\ref{thm:FPT_gen}). While it is easy to see that the $\ODP_\Fset$ problem is decidable for any finite field $\Fset$, the real case relies on a result of Tarski~\cite{Tarski51} on the decidability of the existential theory of the reals (see Proposition~\ref{prop:OD_R_PSPACE}).

We next consider the kernelizability of the $d\dODP_\Fset$ problem parameterized by the vertex cover number for a fixed integer $d$.
For finite fields $\Fset$, one may deduce from a result of~\cite{JansenP19color} that the problem admits a kernel of polynomial size, where the degree of the polynomial grows exponentially with $d$. We prove the following generalized and stronger result.

\begin{theorem}\label{thm:IntroUpperF}
For every field $\Fset$ and for every integer $d \geq 3$, the $d\dODP_\Fset$ problem on $\Empty+k\mathrm{v}$ graphs admits a kernel with $O(k^d)$ vertices and bit-size $O(k^d)$.
\end{theorem}

Theorem~\ref{thm:IntroUpperF} prompts us to determine the smallest possible kernel size for the $d\dODP_\Fset$ problem on $\Empty+k\mathrm{v}$ graphs.
The following result furnishes a lower bound, conditioned on the complexity-theoretic assumption $\NP \nsubseteq \coNPpoly$.
\begin{theorem}\label{thm:IntroLower}
For every field $\Fset$, every integer $d \geq 3$, and any real $\eps>0$, the $d\dODP_\Fset$ problem on $\Empty+k\mathrm{v}$ graphs does not admit a kernel with bit-size $O(k^{d-1-\eps})$ unless $\NP \subseteq \coNPpoly$.
\end{theorem}
\noindent
The proof of Theorem~\ref{thm:IntroLower} combines the lower bound on kernels for coloring problems proved in~\cite{JansenK13,JansenP19color} with a novel linear-parameter transformation from those problems to those associated with orthogonality dimension.
More specifically, we show that for every field $\Fset$ and for every integer $d \geq 3$, it is possible to efficiently transform a graph $G$ into a graph $G'$ so that $\chi(G) \leq d$ if and only if $\od_\Fset(G') \leq d$ while essentially preserving the vertex cover number (see Theorem~\ref{thm:reductionColOD}).
This is in contrast to a reduction of~\cite{Peeters96}, which is appropriate only for $d=3$ and significantly increases the vertex cover number.
The transformation relies on a gadget graph that enforces the vectors assigned to two specified vertices to be either orthogonal or equal up to scalar multiplication (see Lemma~\ref{lemma:gadget}).

We remark that Theorem~\ref{thm:IntroLower} implies that, unless $\NP \subseteq \coNPpoly$, the degree of the polynomial lower bound on the size of a kernel for the $d\dODP_\Fset$ problem on $\Empty+k\mathrm{v}$ graphs can be arbitrarily large when $d$ grows. This yields that the $\ODP_\Fset$ problem on $\Empty+k\mathrm{v}$ graphs, in which $d$ constitutes part of the input, is unlikely to admit a kernel of polynomial size.

Theorems~\ref{thm:IntroUpperF} and~\ref{thm:IntroLower} leave a multiplicative gap of roughly $k$ between the upper and lower bounds on the kernel size achievable for the $d\dODP_\Fset$ problem parameterized by the vertex cover number. For the real field $\R$, we narrow this gap to a multiplicative term of $k^{o(1)}$, as stated below.

\begin{theorem}\label{thm:IntroUpperR}
For every integer $d \geq 3$, the $d\dODP_\R$ problem on $\Empty+k\mathrm{v}$ graphs admits a kernel with $O(k^{d-1})$ vertices and bit-size $O(k^{d-1} \cdot \log k)$.
\end{theorem}
\noindent
The proof of Theorem~\ref{thm:IntroUpperR} borrows the sparsification technique of~\cite{JansenP19sparse} (see also~\cite{JansenP19color}).
A key technical ingredient in applying this method lies in a construction of a low-degree polynomial, which assesses the feasibility of assigning a vector to a vertex based on the vectors of its neighbors (see Lemma~\ref{lemma:det}). Our construction hinges on the fact that the zero vector is the only self-orthogonal vector over the reals. It would be interesting to decide whether or not a similar upper bound on the kernel size could be obtained for finite fields, where this property does not hold.

We finally turn to the study of kernels for the $d\dODP_\Fset$ problem on $\calG+k\mathrm{v}$ graphs for general hereditary graph families $\calG$.
Our first result in this context offers a sufficient condition on $\calG$ for the existence of a polynomial size kernel for $d\dODP_\Fset$ on $\calG+k\mathrm{v}$ graphs.
This condition is related to a variant of the $d\dODP_\Fset$ problem, termed $d\dSubChoose_\Fset$, which was previously studied in various forms (see, e.g.,~\cite{HaynesPSWM10,ChawinH22}) and may be viewed as a counterpart of the $q\qListCol$ problem for orthogonal representations (see Definition~\ref{def:SC}). In this problem, the input consists of a graph $G$ and an assignment of a subspace of $\Fset^d$ to each vertex, and the goal is to decide whether $G$ admits an orthogonal representation over $\Fset$ that assigns to every vertex a vector from its subspace. We show that if the $d\dSubChoose_\Fset$ problem on graphs from a family $\calG$ is local, in the sense that every $\NO$ instance has a $\NO$ sub-instance on at most $g(d)$ vertices, then the $d\dODP_\Fset$ problem on $\calG+k\mathrm{v}$ graphs admits a kernel with $O(k^{d \cdot g(d)})$ vertices. For the precise statement, see Theorem~\ref{thm:KernelGenG}. We demonstrate the applicability of this result for the graph families $\USplit$ and $\UCochordal$.
On the contrary, for the $\Path$ family, we show that it is unlikely that the $d\dODP_\Fset$ problem on $\Path+k\mathrm{v}$ graphs admits a polynomial size kernel even for $d=3$ (see Theorem~\ref{thm:ODPath}).

\subsection{Outline}
The remainder of the paper is structured as follows.
In Section~\ref{sec:preliminaries}, we collect several definitions and facts that will be used throughout the paper.
In Section~\ref{sec:FPT}, we study the fixed-parameter tractability of the $\ODP_\Fset$ problem parameterized by the vertex cover number and prove Theorem~\ref{thm:IntroFPT}.
In Section~\ref{sec:KernelVC}, we present polynomial size kernels for the $d\dODP_\Fset$ problem parameterized by the vertex cover number and prove Theorems~\ref{thm:IntroUpperF} and~\ref{thm:IntroUpperR}. In Section~\ref{sec:lower}, we complement the results of Section~\ref{sec:KernelVC} by providing limits on the kernelizability of the $d\dODP_\Fset$ problem parameterized by the vertex cover number and prove Theorem~\ref{thm:IntroLower}. We further show there that the $3\dODP_\Fset$ problem on $\Path+k\mathrm{v}$ graphs is unlikely to admit a kernel of polynomial size. Finally, in Section~\ref{sec:KernelG}, we study the kernelizability of the $d\dODP_\Fset$ problem on $\calG+k\mathrm{v}$ graphs for general hereditary graph families $\calG$.

\section{Preliminaries}\label{sec:preliminaries}

\subsection{Notations}
For an integer $n$, let $[n] = \{1,2,\ldots,n\}$.
All graphs considered in this paper are simple.
For a graph $G=(V,E)$ and a set $X \subseteq V$, we let $G[X]$ denote the subgraph of $G$ induced by $X$.
The set $X$ is called a vertex cover of $G$ if every edge of $G$ is incident with a vertex of $X$.
We let $G \setminus X$ denote the graph obtained from $G$ by removing the vertices of $X$ (and the edges that touch them).
For a vertex $v \in V$, we let $N_G(v)$ denote the set of neighbors of $v$ in $G$.
A family of graphs is called hereditary if it is closed under vertex removal, or equivalently, under taking induced subgraphs.

\subsection{Linear Algebra}

For a field $\Fset$ and an integer $d$, two vectors $x,y \in \Fset^d$ are said to be orthogonal if $\langle x,y \rangle = 0$, where $\langle x,y \rangle = \sum_{i=1}^{d}{x_i y_i}$ with operations over $\Fset$. If $\langle x,x \rangle = 0$ then $x$ is self-orthogonal, and otherwise it is non-self-orthogonal.
The orthogonal complement of a subspace $W \subseteq \Fset^d$ is the subspace $W^{\perp}$ of all vectors in $\Fset^d$ that are orthogonal to all vectors of $W$, that is, $W^{\perp} = \{ x \in \Fset^d \mid \forall y \in W,~\langle x,y \rangle =0\}$. Note that the orthogonal complement satisfies $\dim(W) + \dim(W^{\perp})=d$ and $W = (W^{\perp})^{\perp}$.
For a vector $u \in \Fset^d$, we let $u^\perp$ denote the orthogonal complement of the subspace $\linspan(u)$ spanned by $u$.
The following simple lemma characterizes the subspaces whose orthogonal complement includes a non-self-orthogonal vector.
Recall that the characteristic of a field is the smallest positive number of copies of the field's identity element that sum to zero, or $0$ if no such number exists.

\begin{lemma}\label{lemma:W_ortho}
Let $\Fset$ be a field, let $d$ be an integer, and let $W$ be a subspace of $\Fset^d$.
\begin{enumerate}
  \item\label{itm:W1} If the characteristic of $\Fset$ is $2$, then there exists a non-self-orthogonal vector in $W^\perp$ if and only if the all-one vector does not lie in $W$.
  \item\label{itm:W2} If the characteristic of $\Fset$ is not $2$, then there exists a non-self-orthogonal vector in $W^\perp$ if and only if $W^\perp \nsubseteq W$.
\end{enumerate}
\end{lemma}

\begin{proof}
We start with the proof of Item~\ref{itm:W1}. Suppose that $\Fset$ has characteristic $2$. It follows that for every vector $x \in \Fset^d$, it holds that $\langle x,x \rangle = \sum_{i=1}^{d}{x_i^2} = (\sum_{i=1}^{d}{x_i})^2$, hence $x$ is self-orthogonal if and only if it is orthogonal to the all-one vector.
Now, if there exists a non-self-orthogonal vector $x \in W^\perp$, then this vector is not orthogonal to the all-one vector. Since $x$ is orthogonal to $W$, it follows that the all-one vector does not lie in $W$.
Conversely, if the all-one vector does not lie in $W = (W^{\perp})^{\perp}$, then there exists a vector in $W^\perp$ that is not orthogonal to the all-one vector. This implies that there exists a non-self-orthogonal vector in $W^\perp$, as required.

We proceed with the proof of Item~\ref{itm:W2}. Suppose that $\Fset$ has characteristic other than $2$.
If there exists a non-self-orthogonal vector in $W^\perp$, then this vector does not lie in $W$, and thus $W^\perp \nsubseteq W$.
Conversely, if  $W^\perp \nsubseteq W$ then there exists a vector $x \in W^\perp$ such that $x \notin W = (W^{\perp})^{\perp}$, hence there exists a vector $y \in W^\perp$ such that $\langle x,y \rangle \neq 0$.
If $\langle x,x \rangle \neq 0$ or $\langle y,y \rangle \neq 0$, then one of $x$ and $y$ is a non-self-orthogonal vector in $W^{\perp}$.
Otherwise, the vector $x+y \in W^\perp$ satisfies this property, because $\langle x+y,x+y \rangle = 2 \cdot \langle x,y \rangle \neq 0$, where for the inequality we use the fact that the characteristic of $\Fset$ is not $2$. This completes the proof.
\end{proof}

Borrowing the terminology of~\cite{JansenP19sparse}, we say that a field $\Fset$ is efficient if field operations and Gaussian elimination can be performed in polynomial time in the size of a reasonable input encoding. All finite fields, as well as the real field $\R$ when restricted to rational numbers (to ensure finite representation), are efficient.

\begin{lemma}\label{lemma:poly}
For every efficient field $\Fset$, there exists a polynomial-time algorithm that given a collection of vectors in $\Fset^d$, decides whether there exists a non-self-orthogonal vector in $\Fset^d$ that is orthogonal to all of them.
\end{lemma}
\begin{proof}
For input vectors $u_1, \ldots, u_\ell \in \Fset^d$, let $W = \linspan(u_1, \ldots, u_\ell)$.
Observe that there exists a non-self-orthogonal vector in $\Fset^d$ that is orthogonal to $u_1, \ldots, u_\ell$ if and only if there exists a non-self-orthogonal vector in the orthogonal complement $W^{\perp}$. By Lemma~\ref{lemma:W_ortho}, for a field $\Fset$ of characteristic $2$, this is equivalent to the all-one vector not lying in $W$, and for every other field $\Fset$, this is equivalent to $W^\perp \nsubseteq W$ (that is, at least one vector of a basis of $W^\perp$ does not lie in $W$).
Since $\Fset$ is an efficient field, these conditions can be checked in polynomial time. This completes the proof.
\end{proof}

For a field $\Fset$ and an integer $d$, if $W$ is a subspace of $\Fset^d$ of dimension smaller than $d$, then its orthogonal complement $W^{\perp}$ has dimension at least $1$, hence there exists a nonzero vector orthogonal to $W$. However, if we require this vector not only to be nonzero but also non-self-orthogonal, its existence is no longer guaranteed. This consideration motivates the following definition.

\begin{definition}\label{def:m(F,d)}
For a field $\Fset$ and an integer $d$, let $m(\Fset,d)$ denote the largest integer $m$ such that for every subspace $W$ of $\Fset^d$ with $\dim(W) < m$, there exists a non-self-orthogonal vector in $W^\perp$.
\end{definition}

\begin{remark}\label{remark:m}
For every field $\Fset$ and for every integer $d \geq 1$, it holds that $1 \leq m(\Fset,d) \leq d$.
Indeed, the lower bound holds because there exists a non-self-orthogonal vector orthogonal to the zero subspace of $\Fset^d$, and the upper bound holds because no nonzero vector is orthogonal to the entire vector space $\Fset^d$.
For a field $\Fset$ of characteristic $2$, it holds that $m(\Fset,d)=1$, because no non-self-orthogonal vector is orthogonal to the $1$-dimensional subspace spanned by the all-one vector.
For every other field $\Fset$, it holds that $m(\Fset,d) \geq \lceil d/2 \rceil$.
To see this, consider a subspace $W \subseteq \Fset^d$ of dimension $\dim(W) < \lceil d/2 \rceil$, and observe that it satisfies $\dim(W^{\perp}) = d-\dim(W) > d-\lceil d/2 \rceil = \lfloor d/2 \rfloor$, and thus $\dim(W^{\perp}) > \dim(W)$. This implies that $W^{\perp} \nsubseteq W$, hence by Item~\ref{itm:W2} of Lemma~\ref{lemma:W_ortho}, there exists a non-self-orthogonal vector in $W^\perp$, as required.
We also observe that if the vector space $\Fset^d$ has no nonzero self-orthogonal vectors, then $m(\Fset,d)=d$, because for every subspace $W \subseteq \Fset^d$ of dimension smaller than $d$ there exists a nonzero vector in $W^\perp$. In particular, for every integer $d$, it holds that $m(\R,d)=d$.
\end{remark}

\subsection{Orthogonality Dimension}

The orthogonality dimension of a graph over a given field is defined as follows.
\begin{definition}\label{def:OD}
For a field $\Fset$ and an integer $d$, a $d$-dimensional orthogonal representation of a graph $G=(V,E)$ over $\Fset$ is an assignment of a vector $u_v \in \mathbb{F}^d$ with $\langle u_v,u_v \rangle \neq 0$ to each vertex $v \in V$, such that for every two adjacent vertices $v$ and $v'$ in $G$, it holds that $\langle u_{v} , u_{v'} \rangle = 0$.
The orthogonality dimension of a graph $G$ over a field $\Fset$, denoted by $\od_\Fset(G)$, is the smallest integer $d$ for which $G$ admits a $d$-dimensional orthogonal representation over $\Fset$.
\end{definition}

\begin{remark}\label{remark:OD}
Let us emphasize that the definition of an orthogonal representation does not require vectors assigned to non-adjacent vertices to be non-orthogonal.
Orthogonal representations that satisfy this additional property are called faithful (see, e.g.,~\cite[Chapter~10]{LovaszBook}).
Note that orthogonal representations of graphs are sometimes defined in the literature as orthogonal representations of the complement, requiring vectors associated with non-adjacent vertices to be orthogonal (with no constraint imposed on vectors of adjacent vertices). We decided to use here the other definition, but one may view the notation $\od_\Fset(G)$ as standing for $\xi_\Fset(\overline{G})$.
\end{remark}

\begin{claim}\label{claim:OD<=CHI}
For every field $\Fset$ and for every graph $G$, it holds that $\od_\Fset(G) \leq \chi(G)$.
\end{claim}
\begin{proof}
For a graph $G=(V,E)$, let $q = \chi(G)$, and consider a proper coloring $c:V \rightarrow [q]$ of $G$.
Assign to each vertex $v \in V$ the vector $e_{c(v)}$ in $\Fset^q$, where $e_i$ stands for the vector of $\Fset^q$ with $1$ on the $i$th entry and $0$ everywhere else. The vectors assigned here to the vertices of $G$ are obviously non-self-orthogonal vectors of $\Fset^q$. Further, for every two adjacent vertices $v$ and $v'$ in $G$, it holds that $c(v) \neq c(v')$, hence $\langle e_{c(v)}, e_{c(v')} \rangle =0$. This implies that there exists a $q$-dimensional orthogonal representation of $G$ over $\Fset$, and thus $\od_\Fset(G) \leq q$.
\end{proof}

\subsection{Parameterized Complexity}\label{sec:parameter}

We present here a few fundamental definitions from the area of parameterized complexity.
For a thorough introduction to the field, the reader is referred to, e.g.,~\cite{CyganFKLMPPS15,KernelBook19}.

A parameterized problem is a set $Q \subseteq \Sigma^* \times \N$ for some finite alphabet $\Sigma$.
A fixed-parameter algorithm for $Q$ is an algorithm that given an instance $(x,k) \in \Sigma^* \times \N$ decides whether $(x,k) \in Q$ in time $f(k) \cdot |x|^c$ for some computable function $f$ and some constant $c$. If $Q$ admits a fixed-parameter algorithm, then we say that $Q$ is fixed-parameter tractable.

A compression (also known as generalized kernel and bikernel) for a parameterized problem $Q \subseteq \Sigma^* \times \N$ into a parameterized problem $Q' \subseteq \Sigma^* \times \N$ is an algorithm that given an instance $(x,k) \in \Sigma^* \times \N$ returns in time polynomial in $|x|+k$ an instance $(x',k') \in \Sigma^* \times \N$, such that $(x,k) \in Q$ if and only if $(x',k') \in Q'$, and in addition, $|x'|+k' \leq h(k)$ for some computable function $h$. The function $h$ is referred to as the size of the compression. If $h$ is polynomial, then the compression is called a polynomial compression. If $|\Sigma|=2$, the function $h$ is called the bit-size of the compression.
When we say that a parameterized problem $Q$ admits a compression of size $h$, we mean that there exists a compression of size $h$ for $Q$ into some parameterized problem.
A compression for a parameterized problem $Q$ into itself is called a kernelization for $Q$ (or simply a kernel).
It is well known that a decidable problem admits a kernel if and only if it is fixed-parameter tractable.

A transformation from a parameterized problem $Q \subseteq \Sigma^* \times \N$ into a parameterized problem $Q' \subseteq \Sigma^* \times \N$ is an algorithm that given an instance $(x,k) \in \Sigma^* \times \N$ returns in time polynomial in $|x|+k$ an instance $(x',k') \in \Sigma^* \times \N$, such that $(x,k) \in Q$ if and only if $(x',k') \in Q'$, and in addition, $k' \leq h(k)$ for some computable function $h$. If $h$ is polynomial, the transformation is called polynomial-parameter, and if $h$ is linear, the transformation is called linear-parameter.

\section{Fixed-Parameter Tractability of \texorpdfstring{$\ODP_\Fset$}{Ortho-Dim F}}\label{sec:FPT}

In this section, we prove that the $\ODP_\Fset$ problem parameterized by the vertex cover number is fixed-parameter tractable for various fields $\Fset$ (recall Definition~\ref{def:OD_P}).

\subsection{Finite Fields}

We begin with the simple case, where the field $\Fset$ is finite.
The proof resembles the one of the fixed-parameter tractability of the $\Col$ problem parameterized by the vertex cover number.

\begin{theorem}\label{thm:FPT_finite}
For every finite field $\Fset$, $\ODP_\Fset$ on $\Empty+k\mathrm{v}$ graphs is fixed-parameter tractable.
\end{theorem}

\begin{proof}
Fix a finite field $\Fset$.
The input of $\ODP_\Fset$ on $\Empty+k\mathrm{v}$ graphs consists of a graph $G=(V,E)$, a vertex cover $X \subseteq V$ of $G$ of size $|X|=k$, and an integer $d$.
Consider the algorithm that given such an input acts as follows.
If $d > k$ then the algorithm accepts.
Otherwise, the algorithm enumerates all possible assignments of non-self-orthogonal vectors from $\Fset^d$ to the vertices of $X$. For every such assignment, the algorithm checks for every vertex $v \in V \setminus X$ if there exists a non-self-orthogonal vector in $\Fset^d$ that is orthogonal to the vectors assigned to the neighbors of $v$ (note that they all lie in $X$). If there exists an assignment to the vertices of $X$ such that the answer is positive for all the vertices of $V \setminus X$, then the algorithm accepts, and otherwise it rejects.

For correctness, observe first that the input graph $G$ is $(k+1)$-colorable, as follows by assigning $k$ distinct colors to the vertices of the vertex cover $X$ and another color to the vertices of the independent set $V \setminus X$. It thus follows, using Claim~\ref{claim:OD<=CHI}, that $\od_\Fset(G) \leq \chi(G) \leq k+1$. Therefore, if $d >k$, then it holds that $\od_\Fset(G) \leq d$, hence our algorithm correctly accepts. Otherwise, the algorithm tries all possible assignments of non-self-orthogonal vectors of $\Fset^d$ to the vertices of $X$. Since the vertices of $V \setminus X$ form an independent set in $G$, an assignment to the vertices of $X$ can be extended to the whole graph if and only if for each vertex $v \in V \setminus X$ there exists a non-self-orthogonal vector in $\Fset^d$ that is orthogonal to the vectors assigned to the neighbors of $v$ (which all lie in $X$). Since this condition is checked by the algorithm for all possible assignments to the vertices of $X$, its answer is correct.

We finally analyze the running time of the algorithm.
On instances with $d>k$, the algorithm is clearly efficient.
For instances with $d \leq k$, the number of assignments of vectors from $\Fset^d$ to the vertices of $X$ is at most ${|\Fset|}^{d \cdot |X|} \leq |\Fset|^{k^2}$.
Further, by Lemma~\ref{lemma:poly}, given a collection of vectors of $\Fset^d$, it is possible to decide in polynomial time whether there exists a non-self-orthogonal vector in $\Fset^d$ that is orthogonal to all of them. This implies that our algorithm for $\ODP_\Fset$ on $\Empty+k\mathrm{v}$ graphs can be implemented in time $|\Fset|^{k^2} \cdot n^{O(1)}$, where $n$ stands for the input size, hence the problem is fixed-parameter tractable.
\end{proof}

\subsection{General Fields}

We turn to the following generalization of Theorem~\ref{thm:FPT_finite}.

\begin{theorem}\label{thm:FPT_gen}
Let $\Fset$ be a field for which the $\ODP_\Fset$ problem is decidable.
Then the $\ODP_\Fset$ problem on $\Empty+k\mathrm{v}$ graphs is fixed-parameter tractable.
\end{theorem}

Recall that the algorithm of Theorem~\ref{thm:FPT_finite} for the $\ODP_\Fset$ problem on $\Empty+k\mathrm{v}$ graphs enumerates all possible assignments of non-self-orthogonal vectors to the vertices of a given vertex cover. This approach is clearly not applicable when the field $\Fset$ is infinite. In order to extend the fixed-parameter tractability result to general fields and to obtain Theorem~\ref{thm:FPT_gen}, we use the following definition inspired by an idea of~\cite{JansenK13}.

\begin{definition}\label{def:K}
Let $G=(V,E)$ be a graph, let $X \subseteq V$ be a vertex cover of $G$, and let $d \geq m \geq 1$ be integers.
We define the graph $\KG = \KG(G,X,m,d)$ as follows.
We start with $\KG = G[X]$. Then, for every subset $S \subseteq X$ of size $m \leq |S| \leq d$, if there exists a vertex $v \in V \setminus X$ such that $S \subseteq N_G(v)$, then we add to $\KG$ a new vertex $v_S$ and connect it to all the vertices of $S$.
\end{definition}

The following lemma lists useful properties of the graph given in Definition~\ref{def:K} (recall Definition~\ref{def:m(F,d)}).

\begin{lemma}\label{lemma:K}
Let $G=(V,E)$ be a graph, let $X \subseteq V$ be a vertex cover of $G$ of size $|X|=k$, let $d \geq m \geq 1$ be integers, and let $\KG = \KG(G,X,m,d)$.
\begin{enumerate}
  \item\label{itm:K1} The set $X$ forms a vertex cover of $\KG$.
  \item\label{itm:K2} The number of vertices in $\KG$ is at most $k + \sum_{i=m}^{d}{\binom{k}{i}}$.
  \item\label{itm:K3} The graph $\KG$ can be encoded in $\binom{k}{2} + \sum_{i=m}^{d}{\binom{k}{i}}$ bits.
  \item\label{itm:K4} For every field $\Fset$ with $m \leq m(\Fset,d)$, it holds that $\od_\Fset(G) \leq d$ if and only if $\od_\Fset(\KG) \leq d$.
\end{enumerate}
\end{lemma}

\begin{proof}
Consider the graph $\KG = \KG(G,X,m,d)$ given in Definition~\ref{def:K}.
Since $X$ is a vertex cover of $G$, it immediately follows from the definition that every edge of $\KG$ is incident with a vertex of $X$, hence $X$ is a vertex cover of $\KG$, as required for Item~\ref{itm:K1}.
It further follows that the vertex set of $\KG$ consists of the vertices of $X$ and at most one vertex per every subset $S \subseteq X$ of size $m \leq |S| \leq d$. Since the number of those subsets is $\sum_{i=m}^{d}{\binom{k}{i}}$, the number of vertices in $\KG$ is at most $k + \sum_{i=m}^{d}{\binom{k}{i}}$, as required for Item~\ref{itm:K2}.
For Item~\ref{itm:K3}, notice that to encode the graph $\KG$, it suffices to specify the adjacencies in $\KG[X]$ and the existence of the vertex $v_S$ in $\KG$ for each $S \subseteq X$ of size $m \leq |S| \leq d$, hence $\KG$ can be encoded in $\binom{k}{2} + \sum_{i=m}^{d}{\binom{k}{i}}$ bits.

We turn to the proof of Item~\ref{itm:K4}.
Let $\Fset$ be a field with $m \leq m(\Fset,d)$.
Suppose first that $\od_\Fset(G) \leq d$, that is, there exists a $d$-dimensional orthogonal representation $(u_v)_{v \in V}$ of $G$ over $\Fset$.
We define a $d$-dimensional orthogonal representation of $\KG$ over $\Fset$ as follows.
First, we assign to each vertex $v \in X$ the vector $u_v$.
It clearly holds that every two vertices of $X$ that are adjacent in $\KG$ are assigned orthogonal vectors.
Next, for each vertex $v_S$ of $\KG$ with $S \subseteq X$ and $m \leq |S| \leq d$, there exists a vertex $v \in V \setminus X$ such that $S \subseteq N_G(v)$. We assign to $v_S$ the vector $u_v$ of such a vertex $v$. Notice that such a vector is orthogonal to all the vectors associated with the vertices of $S$, i.e., the neighbors of $v_S$ in $\KG$. This gives us a $d$-dimensional orthogonal representation of $\KG$ over $\Fset$, implying that $\od_\Fset(\KG) \leq d$.

For the other direction, suppose that $\od_\Fset(\KG) \leq d$. Letting $V'$ denote the vertex set of $\KG$, there exists a $d$-dimensional orthogonal representation $(u_v)_{v \in V'}$ of $\KG$ over $\Fset$.
We define a $d$-dimensional orthogonal representation of $G$ over $\Fset$ as follows.
First, we assign to each vertex $v \in X$ the vector $u_v$.
It clearly holds that every two vertices of $X$ that are adjacent in $G$ are assigned orthogonal vectors.
We next extend this assignment to the vertices of the independent set $V \setminus X$ of $G$.
Consider some vertex $v \in V \setminus X$, let $W = \linspan(\{u_{v'} \mid v' \in N_G(v)\})$, and notice that $v$ may be assigned any non-self-orthogonal vector of $\Fset^d$ that lies in $W^\perp$.
If $\dim(W) < m$, then by $m \leq m(\Fset,d)$, there exists a non-self-orthogonal vector in $W^\perp$, which can be assigned to the vertex $v$.
Otherwise, there exists a set of vertices $S \subseteq N_G(v)$ of size $m \leq |S| \leq d$ whose vectors form a basis of $W$, that is, $W = \linspan(\{ u_{v'} \mid v' \in S\})$. By the definition of the graph $\KG$, it includes the vertex $v_S$, and its vector is orthogonal to the vectors $u_{v'}$ with $v' \in S$, and thus lies in $W^\perp$. This yields the existence of the desired vector for $v$, so we are done.
\end{proof}

With Lemma~\ref{lemma:K} at hand, we are ready to prove Theorem~\ref{thm:FPT_gen}.

\begin{proof}[ of Theorem~\ref{thm:FPT_gen}]
The input of the $\ODP_\Fset$ problem on $\Empty+k\mathrm{v}$ graphs consists of a graph $G$, a vertex cover $X$ of $G$ of size $|X|=k$, and an integer $d$.
Consider the algorithm that given such an input acts as follows.
If $d > k$ then the algorithm accepts.
Otherwise, the algorithm calls an algorithm for the $\ODP_\Fset$ problem on the input $(\KG,d)$, where $\KG=\KG(G,X,1,d)$ is the graph given in Definition~\ref{def:K}, and returns its answer.
Note that we use here the assumption that the $\ODP_\Fset$ problem is decidable.

For correctness, observe first that the input graph $G$ is $(k+1)$-colorable, as follows by assigning $k$ distinct colors to the vertices of the vertex cover $X$ and another color to the vertices of the independent set $V \setminus X$. It thus follows, using Claim~\ref{claim:OD<=CHI}, that $\od_\Fset(G) \leq \chi(G) \leq k+1$. Therefore, if $d >k$, then it holds that $\od_\Fset(G) \leq d$, hence our algorithm correctly accepts. Otherwise, the algorithm calls an algorithm for $\ODP_\Fset$ on the input $(\KG,d)$. The correctness of its answer follows from Item~\ref{itm:K4} of Lemma~\ref{lemma:K}, which guarantees that $\od_\Fset(G) \leq d$ if and only if $\od_\Fset(\KG) \leq d$.

We finally analyze the running time of the algorithm.
On instances with $d>k$, the algorithm is clearly efficient.
For instances with $d \leq k$, by Item~\ref{itm:K2} of Lemma~\ref{lemma:K}, the number of vertices in $\KG$ is $O(k^d) \leq O(k^k)$.
Using the decidability of $\ODP_\Fset$, this implies that the running time of the algorithm is bounded by $f(k) \cdot n^{O(1)}$ for some computable function $f$, where $n$ stands for the input size. Therefore, the $\ODP_\Fset$ problem on $\Empty+k\mathrm{v}$ graphs is fixed-parameter tractable.
\end{proof}

In order to apply Theorem~\ref{thm:FPT_gen} to the real field $\R$, one has to show that the $\ODP_\R$ problem is decidable.
We obtain this result using the problem of the existential theory of the reals, in which the input is a collection of equalities and inequalities of polynomials over the reals, and the goal is to decide whether there exists an assignment of real values to the variables satisfying all the constraints. In 1951, Tarski~\cite{Tarski51} proved that the problem is decidable. His result was strengthened in 1988 by Canny~\cite{Canny88}, who proved that it actually lies in the complexity class $\PSPACE$.
We derive the following simple consequence.

\begin{proposition}\label{prop:OD_R_PSPACE}
The $\ODP_\R$ problem lies in $\PSPACE$.
\end{proposition}

\begin{proof}
It is sufficient to show that the $\ODP_\R$ problem is reducible in polynomial time to the problem of the existential theory of the reals, which lies in $\PSPACE$~\cite{Canny88}.
Consider the reduction that given a graph $G = (V,E)$ and an integer $d$ produces a collection $P_G$ of polynomial constraints over the reals defined as follows.
For each vertex $v \in V$, let $x_{v,1}, \ldots, x_{v,d}$ denote $d$ variables associated with $v$.
For each vertex $v \in V$, add to $P_G$ the inequality $\sum_{i=1}^{d}{x_{v,i}^2} \neq 0$, and for each edge $\{v,v'\} \in E$, add to $P_G$ the equality $\sum_{i=1}^d{x_{v,i} \cdot x_{v',i}} = 0$. The reduction returns the collection $P_G$, which can clearly be computed in polynomial time.
Observe that $\od_\R(G) \leq d$ if and only if there exists an assignment over the reals satisfying the constraints of $P_G$, implying the correctness of the reduction.
\end{proof}

Proposition~\ref{prop:OD_R_PSPACE} implies that the $\ODP_\R$ problem is decidable.
Using Theorem~\ref{thm:FPT_gen}, we obtain the following corollary, which combined with Theorem~\ref{thm:FPT_finite}, confirms Theorem~\ref{thm:IntroFPT}.
\begin{corollary}\label{cor:FPT_R}
The $\ODP_\R$ problem on $\Empty+k\mathrm{v}$ graphs is fixed-parameter tractable.
\end{corollary}

\section{Kernelization for \texorpdfstring{$d\dODP_\Fset$}{d-Ortho-Dim F} Parameterized by Vertex Cover}\label{sec:KernelVC}

We consider now the $d\dODP_\Fset$ problem for a fixed constant $d$ and study its kernelizability when parameterized by the vertex cover number (recall Definition~\ref{def:OD_P}).
We first leverage our discussion from the previous section to derive Theorem~\ref{thm:IntroUpperF}, namely, to show that for every field $\Fset$ and for every integer $d \geq 3$, the $d\dODP_\Fset$ problem on $\Empty+k\mathrm{v}$ graphs admits a kernel with $O(k^d)$ vertices and bit-size $O(k^d)$.

\begin{proof}[ of Theorem~\ref{thm:IntroUpperF}]
Fix a field $\Fset$ and an integer $d \geq 3$.
The input of $d\dODP_\Fset$ on $\Empty+k\mathrm{v}$ graphs consists of a graph $G$ and a vertex cover $X$ of $G$ of size $|X|=k$.
Consider the algorithm that given such an input returns the pair $(\KG,X)$, where $\KG = \KG(G,X,1,d)$ is the graph from Definition~\ref{def:K}.
Since $d$ is a fixed constant, the algorithm can be implemented in polynomial time.
By Lemma~\ref{lemma:K}, the set $X$ forms a vertex cover of $\KG$, the graph $\KG$ has $O(k^d)$ vertices and bit-size $O(k^d)$, and the instances $(G,X)$ and $(\KG,X)$ are equivalent. This completes the proof.
\end{proof}

For the real field $\R$, we prove Theorem~\ref{thm:IntroUpperR}, which improves on the kernel provided by Theorem~\ref{thm:IntroUpperF} to $O(k^{d-1})$ vertices and bit-size $O(k^{d-1} \cdot \log k)$. We start with a couple of auxiliary lemmas.

\begin{lemma}\label{lemma:1}
For every integer $d$, if a graph has a $d$-dimensional orthogonal representation over $\R$, then it has a $d$-dimensional orthogonal representation over $\R$, all of whose vectors have $1$ as their first entry.
\end{lemma}

\begin{proof}
The proof applies the probabilistic method.
Let $d$ be an integer, let $G=(V,E)$ be a graph, and set $n=|V|$.
Suppose that there exists a $d$-dimensional orthogonal representation $(u_v)_{v \in V}$ of $G$ over $\R$.
Let $a \in [2n]^d$ be a random $d$-dimensional vector, such that each entry of $a$ is chosen from $[2n]$ uniformly at random.
We observe that for every fixed nonzero vector $u \in \R^d$, it holds that $\langle a,u \rangle = 0$ with probability at most $\frac{1}{2n}$.
Indeed, letting $i \in [d]$ be an index with $u_i \neq 0$, for every fixed choice of the values of $a_j$ with $j \in [d] \setminus \{i\}$, there is at most one value of $a_i$ in $[2n]$ for which it holds that $\langle a,u \rangle = 0$.
By the union bound, it follows that the probability that there exists a vertex $v \in V$ such that $\langle a, u_v \rangle = 0$ is at most $n \cdot \frac{1}{2n} = \frac{1}{2}$. In particular, there exists a vector $a \in [2n]^d$ satisfying $\langle a, u_v \rangle \neq 0$ for all $v \in V$. Let us fix such a vector $a$.

Now, let $M \in \R^{d \times d}$ be some orthonormal matrix (i.e., a matrix satisfying $M \cdot M^t = I_d$) whose first row is the vector $a$ scaled to have Euclidean norm $1$, i.e., $a/\|a\|$.
We assign to each vertex $v \in V$ of the graph $G$ the vector $M \cdot u_v$.
Since $M$ is orthonormal, it preserves inner products, hence this assignment forms a $d$-dimensional orthogonal representation of $G$ over $\R$.
Additionally, for every vertex $v \in V$, the first entry of the vector $M \cdot u_v$ is nonzero, because $\langle a,u_v \rangle \neq 0$.
By scaling, one can obtain a $d$-dimensional orthogonal representation of $G$ over $\R$, all of whose vectors have $1$ as their first entry, as required.
\end{proof}

Before we state the next lemma, we need a brief preparation.
For a field $\Fset$, a polynomial in $\Fset[x_1,\ldots,x_n]$ is called homogeneous of degree $d$ if each of its monomials has degree $d$. Note that the zero polynomial is homogeneous of degree $d$ for every $d \geq 0$.
A monomial is called multilinear if it forms a product of distinct variables, and a polynomial is called multilinear if it forms a linear combination of multilinear monomials.
For example, the determinant of $d \times d$ matrices over a field $\Fset$, viewed as a polynomial on $d^2$ variables, is multilinear and homogeneous of degree $d$.
Moreover, it is a linear combination of $d!$ monomials, each of which forms a product of $d$ variables, one taken from each row of the matrix.
Note that the dimension over $\Fset$ of the vector space of multilinear homogeneous polynomials of degree $d$ in $\Fset[x_1, \ldots, x_n]$ is $\binom{n}{d}$.

\begin{lemma}\label{lemma:det}
For every integer $d$, there exists a multilinear homogeneous polynomial $p: \R^{d \times d} \rightarrow \R$ of degree $d-1$, defined on $d^2$ variables corresponding to the entries of a $d \times d$ matrix, such that for every matrix $M \in \R^{d \times d}$ whose first row is the all-one vector, it holds that $p(M)=0$ if and only if there exists a nonzero vector in $\R^d$ that is orthogonal to all columns of $M$.
\end{lemma}

\begin{proof}
For an integer $d$, consider the determinant polynomial $\det : \R^{d \times d} \rightarrow \R$.
It is well known that for every matrix $M \in \R^{d \times d}$, it holds that $\det(M) = 0$ if and only if the columns of $M$ span a subspace of dimension smaller than $d$, and that this condition is equivalent to the existence of a nonzero vector in $\R^d$ that is orthogonal to all columns of $M$.
Recall that $\det$ is a multilinear polynomial, with each monomial being a product of $d$ variables, each selected from a different row of the matrix.
Let $p: \R^{d \times d} \rightarrow \R$ be the polynomial obtained from $\det$ by substituting $1$ for the variables that correspond to the first row of the matrix, and observe that $p$ is a multilinear homogeneous polynomial of degree $d-1$.
Note that although $p$ is defined on $d^2$ variables, it actually depends on only $d^2-d$ of them.
We finally observe that for every matrix $M \in \R^{d \times d}$ whose first row is the all-one vector, it holds that $p(M)=0$ if and only if there exists a nonzero vector in $\R^d$ that is orthogonal to all columns of $M$.
This completes the proof.
\end{proof}

We are ready to prove Theorem~\ref{thm:IntroUpperR}, providing a kernel with $O(k^{d-1})$ vertices and bit-size $O(k^{d-1} \cdot \log k)$ for the $d\dODP_\R$ problem on $\Empty+k\mathrm{v}$ graphs for all integers $d \geq 3$.

\begin{proof}[ of Theorem~\ref{thm:IntroUpperR}]
Fix an integer $d \geq 3$.
The input of $d\dODP_\R$ on $\Empty+k\mathrm{v}$ graphs consists of a graph $G=(V,E)$ and a vertex cover $X \subseteq V$ of $G$ of size $|X|=k$.
Consider the algorithm that given such an input acts in two phases, as described next.

In the first phase, the algorithm constructs the graph $G' = \KG(G,X,d,d)$ given in Definition~\ref{def:K}.
Let $V'$ denote the vertex set of $G'$, and recall that every vertex $v_S \in V' \setminus X$ is associated with some set $S \subseteq X$ of size $|S|=d$ such that $N_{G'}(v_S)=S$. By Lemma~\ref{lemma:K}, the set $X$ is a vertex cover of $G'$, and it holds that $|V'| \leq k+\binom{k}{d}$.

In the second phase, the algorithm constructs a graph $G''$.
To do so, the algorithm first associates with each vertex $v \in X$ a $d$-dimensional vector $x_v$ of variables over $\R$. Note that the total number of variables is $k \cdot d$.
For each vertex $v_S \in V' \setminus X$, we apply Lemma~\ref{lemma:det} to obtain a multilinear homogeneous polynomial $p_S$ of degree $d-1$, defined on the $d^2$ variables associated with the $d$ neighbors of $v_S$ in $G'$ (which all lie in $X$). The polynomial $p_S$ satisfies that for every assignment $M \in \R^{d \times d}$ to its variables with first row equal to the all-one vector, it holds that $p_S(M)=0$ if and only if there exists a nonzero vector in $\R^d$ that is orthogonal to all columns of $M$.
Let $P = \linspan(\{ p_S ~|~ v_S \in V' \setminus X\})$ denote the subspace spanned by the polynomials associated with the vertices of $V' \setminus X$.
The algorithm proceeds by finding a set $Y \subseteq V' \setminus X$, such that the polynomials associated with the vertices of $Y$ form a basis for $P$. Note that $P$ is contained in the vector space of multilinear homogeneous polynomials of degree $d-1$ on $k \cdot d$ variables. Since the dimension of the latter is $\binom{k \cdot d}{d-1}$, recalling that $d$ is a fixed constant, it follows that $|Y| \leq \binom{k \cdot d}{d-1} \leq O(k^{d-1})$.
Letting $V'' = X \cup Y$, the algorithm returns the graph $G'' = G'[V'']$ and the set $X$, which forms a vertex cover of $G''$ because it forms a vertex cover of $G'$.

The number of vertices in $G''$ is $|V''| = |X|+|Y| \leq k + O(k^{d-1}) = O(k^{d-1})$.
The number of edges in $G''[X]$ is at most $\binom{k}{2}$, and since the degree of each vertex of $Y$ is $d$, the number of edges in $G''$ that involve vertices of $Y$ is $d \cdot |Y|$. It follows that the total number of edges in $G''$ is at most $\binom{k}{2} + d \cdot O(k^{d-1}) \leq O(k^{d-1})$.
Therefore, the number of bits required to encode the edges of $G''$ is at most $O(k^{d-1} \cdot \log |V''|) \leq O(k^{d-1} \cdot \log k)$, as required.

It is not difficult to verify that the algorithm can be implemented in polynomial time.
Note that the set $Y$ can be calculated in polynomial time by applying Gaussian elimination with $\binom{k \cdot d}{d-1}$ variables.

For the correctness of the algorithm, we shall prove that $\od_\R(G) \leq d$ if and only if $\od_\R(G'') \leq d$.
By Item~\ref{itm:K4} of Lemma~\ref{lemma:K}, using $m(\R,d)=d$ (see Remark~\ref{remark:m}), it holds that $\od_\R(G) \leq d$ if and only if $\od_\R(G') \leq d$.
It thus suffices to show that $\od_\R(G') \leq d$ if and only if $\od_\R(G'') \leq d$.

It obviously holds that if $\od_\R(G') \leq d$ then $\od_\R(G'') \leq d$, because $G'$ contains $G''$ as a subgraph.
For the converse, suppose that $\od_\R(G'') \leq d$, that is, there exists a $d$-dimensional orthogonal representation of $G''$ over $\R$.
By Lemma~\ref{lemma:1}, it further follows that there exists a $d$-dimensional orthogonal representation $(u_v)_{v \in V''}$ of $G''$ over $\R$, such that every vector $u_v$ has $1$ as its first entry.
For each vertex $v \in X$, assign the vector $u_v$ to the vertex $v$ as well as to the variables of the vector $x_v$ associated with $v$.
We will show that this assignment to the vertices of $X$ can be extended to an orthogonal representation of $G'$ over $\R$.
Indeed, for every vertex $v_S \in Y$ of $G''$, the nonzero vector $u_{v_S}$ is orthogonal to the vectors of the vertices of $S$. This implies, using Lemma~\ref{lemma:det} and the fact that the first entries of the vectors $x_v$ with $v \in X$ are all $1$, that the polynomial $p_S$ vanishes on this assignment. Since the polynomials $p_S$ with $v_S \in Y$ form a basis of the subspace $P$, it follows that all the polynomials $p_S$ associated with the vertices $v_S \in V' \setminus X$ vanish on this assignment as well. Using Lemma~\ref{lemma:det} again, we obtain that for each vertex $v_S \in V' \setminus X$, there exists a nonzero vector that is orthogonal to the vectors of the vertices of $S$, and these are precisely the neighbors of $v_S$ in $G'$. This gives us a $d$-dimensional orthogonal representation of $G'$ over $\R$, which yields that $\od_\R(G') \leq d$, concluding the proof.
\end{proof}

\section{Lower Bounds}\label{sec:lower}

In this section, we prove our lower bounds on the kernel size of orthogonality dimension problems.
We first present the gadget graph that will be used in the proofs.

\subsection{Gadget Graph}

A key ingredient in the proofs of our lower bounds is the following lemma, which generalizes a construction of~\cite{Peeters96} (corresponding to the case of $d=3$).
Here, two nonzero vectors $u_1,u_2 \in \Fset^d$ are said to be proportional if there exists some $\alpha \in \Fset$ such that $u_1 = \alpha \cdot u_2$.

\begin{lemma}\label{lemma:gadget}
For an integer $d \geq 3$, let $C_{2d}$ denote the cycle graph on $2d$ vertices, let $x_0$ and $x_1$ denote two adjacent vertices in the cycle, and let $H = \overline{C_{2d}}$ denote its complement graph.
\begin{enumerate}
  \item\label{itm:same} There exists a proper $d$-coloring of $H$ that assigns to $x_0$ and $x_1$ the same color.
  \item\label{itm:distinct} There exists a proper $d$-coloring of $H$ that assigns to $x_0$ and $x_1$ distinct colors.
  \item\label{itm:od} For every field $\Fset$ and for every $d$-dimensional orthogonal representation of $H$ over $\Fset$, the vectors assigned to $x_0$ and $x_1$ are either orthogonal or proportional.
\end{enumerate}
\end{lemma}

\begin{proof}
Let $x_0$ and $x_1$ be two adjacent vertices in the cycle graph $C_{2d}$, and let $x_0, x_1, \ldots, x_{2d-1}$ denote its $2d$ vertices ordered along the cycle. Throughout the proof, all the indices $i$ of the vertices $x_i$ are considered modulo $2d$.
For Item~\ref{itm:same}, consider the coloring of $H$ that for each $0 \leq i < d$ assigns the color $i$ to the vertices $x_{2i}$ and $x_{2i+1}$.
Notice that this assignment forms a proper $d$-coloring of $H$ and that it assigns to $x_0$ and $x_1$ the same color $0$.
For Item~\ref{itm:distinct}, consider the coloring of $H$ that for each $0 \leq i < d$ assigns the color $i$ to the vertices $x_{2i-1}$ and $x_{2i}$.
Notice that this assignment forms a proper $d$-coloring of $H$ and that it assigns to $x_0$ and $x_1$ the distinct colors $0$ and $1$ respectively.

For Item~\ref{itm:od}, let $\Fset$ be a field, and consider a $d$-dimensional orthogonal representation of $H$ over $\Fset$ assigning a vector $u_i \in \Fset^d$ to each vertex $x_i$. Note that for all indices $i$ and $j$ with $j \notin \{i-1,i,i+1\}$, the vertices $x_i$ and $x_j$ are adjacent in $H$, hence it holds that $\langle u_i,u_j \rangle = 0$.

We first claim that for each $i$, the vector $u_i$ is a linear combination of the vectors $u_{i-1}$ and $u_{i+1}$.
By symmetry, it suffices to prove it for $i=1$.
The vertices $x_0, x_2, \ldots, x_{2d-2}$ of even indices form a clique of size $d$ in $H$, hence their vectors $u_0, u_2, \ldots, u_{2d-2}$ are pairwise orthogonal. Since they are non-self-orthogonal, it follows that they are linearly independent and thus span the entire vector space $\Fset^d$. Therefore, there exist some coefficients $\alpha_0, \ldots, \alpha_{d-1} \in \Fset$ such that $u_1 = \sum_{0 \leq j < d}{\alpha_j \cdot u_{2j}}$.
It follows that for every integer $k$ with $1 < k < d$, it holds that
\[0 = \langle u_1, u_{2k} \rangle = \sum_{0 \leq j < d}{\alpha_j \cdot \langle u_{2j},u_{2k} \rangle } = \alpha_{k} \cdot \langle u_{2k},u_{2k} \rangle,\]
which using $\langle u_{2k},u_{2k} \rangle \neq 0$ implies that $\alpha_{k} = 0$. We thus obtain that $u_1 = \alpha_0 \cdot u_0 + \alpha_1 \cdot u_2$, so $u_1$ is a linear combination of $u_0$ and $u_2$, as required.

Now, we prove that the vectors $u_0$ and $u_1$, which are assigned to the vertices $x_0$ and $x_1$ respectively, are either orthogonal or proportional.
The above discussion implies that there exist some coefficients $\alpha_0, \alpha_1, \beta_0, \beta_1, \gamma_0, \gamma_1 \in \Fset$ such that
\[u_1 = \alpha_0 \cdot u_0 + \alpha_1 \cdot u_2,~~~u_2 = \beta_0 \cdot u_1 + \beta_1 \cdot u_3,\mbox{~~~and~~~}u_3 = \gamma_0 \cdot u_2 + \gamma_1 \cdot u_4.\]
It follows that \[0 = \langle u_1, u_3 \rangle = \langle \alpha_0 \cdot u_0 + \alpha_1 \cdot u_2, \gamma_0 \cdot u_2 + \gamma_1 \cdot u_4 \rangle = \alpha_1 \cdot \gamma_0 \cdot \langle u_2,u_2 \rangle,\]
which using $\langle u_2,u_2 \rangle \neq 0$ implies that $\alpha_1 = 0$ or $\gamma_0 = 0$.
If $\alpha_1 = 0$, then it holds that $u_1 = \alpha_0 \cdot u_0$, hence $u_0$ and $u_1$ are proportional, and we are done.
Otherwise, it holds that $\gamma_0 = 0$, implying that $u_3 = \gamma_1 \cdot u_4$.
This yields that $u_2 = \beta_0 \cdot u_1 + \beta_1 \cdot \gamma_1 \cdot u_4$, hence $0= \langle u_2,u_4 \rangle = \beta_1 \cdot \gamma_1 \cdot \langle u_4, u_4 \rangle$.
Using $\langle u_4, u_4 \rangle \neq 0$, we deduce that $\beta_1 \cdot \gamma_1 = 0$, hence $u_2 = \beta_0 \cdot u_1$ and $\beta_0 \neq 0$.
Since $u_2$ is orthogonal to $u_0$, it follows that so is $u_1$, and we are done.
\end{proof}

\subsection{The \texorpdfstring{$d\dODP_\Fset$}{d-Ortho-Dim F} Problem on \texorpdfstring{$\Empty+k\mathrm{v}$}{EMPTY+kv} Graphs}

We prove the following lower bound on the size of any compression for the $d\dODP_\Fset$ problem parameterized by the vertex cover number.

\begin{theorem}\label{thm:Lower_d}
For every field $\Fset$, every integer $d \geq 3$, and any real $\eps>0$, the $d\dODP_\Fset$ problem on $\Empty+k\mathrm{v}$ graphs does not admit a compression with bit-size $O(k^{d-1-\eps})$ unless $\NP \subseteq \coNPpoly$.
\end{theorem}

Note that Theorem~\ref{thm:Lower_d} confirms Theorem~\ref{thm:IntroLower}.
Another immediate corollary is the following.
\begin{corollary}
For every field $\Fset$, the $\ODP_\Fset$ problem on $\Empty+k\mathrm{v}$ graphs does not admit a polynomial compression unless $\NP \subseteq \coNPpoly$.
\end{corollary}

The starting point of the proof of Theorem~\ref{thm:Lower_d} is the following theorem, which summarizes the lower bounds proved in~\cite{JansenK13,JansenP19color} on the size of any compression for the $q\qCol$ problem parameterized by the vertex cover number (see~\cite[Corollary~2]{JansenP19color}).

\begin{theorem}[{\cite{JansenK13,JansenP19color}}]\label{thm:NoKernelCol}
For every integer $q \geq 3$, the $q\qCol$ problem on $\Empty+k\mathrm{v}$ graphs does not admit a compression with bit-size $O(k^{q-1-\eps})$ unless $\NP \subseteq \coNPpoly$.
\end{theorem}

Equipped with Lemma~\ref{lemma:gadget}, we relate the $d$-$\Col$ and $d\dODP_\Fset$ problems parameterized by the vertex cover number, as stated below.

\begin{theorem}\label{thm:reductionColOD}
For every field $\Fset$ and for every integer $d \geq 3$, there exists a linear-parameter transformation from $d$-$\Col$ on $\Empty+k\mathrm{v}$ graphs to $d\dODP_\Fset$ on $\Empty+k\mathrm{v}$ graphs.
\end{theorem}

\begin{proof}
Fix a field $\Fset$ and an integer $d \geq 3$.
Consider an instance of the $d$-$\Col$ problem on $\Empty+k\mathrm{v}$ graphs, namely, a graph $G=(V,E)$ and a vertex cover $X \subseteq V$ of $G$ of size $|X|=k$.
Our goal is to construct in polynomial time a graph $G'=(V',E')$ and a vertex cover $X' \subseteq V'$ of $G'$ of size $|X'| = O(k)$, such that $\chi(G) \leq d$ if and only if $\od_\Fset(G') \leq d$.

To do so, we start with the graph $G$ and add to it a clique of size $d$ whose vertices are denoted by $z_1, \ldots, z_d$.
Then, for each index $i \in [d]$ and each vertex $v \in X$, we add to the graph a copy $H_{i,v}$ of the complement $\overline{C_{2d}}$ of the cycle graph on $2d$ vertices, where two consecutive vertices of the cycle are identified with the vertices $z_i$ and $v$. Note that we add here $d \cdot k$ such gadgets to the graph and that each of them involves $2d-2$ new vertices. Let $G' = (V',E')$ denote the obtained graph, and let $X'$ denote the set that consists of the vertices of $X$ and the vertices that were added to $G$ in the construction. The transformation returns the pair $(G',X')$. Since $X$ is a vertex cover of $G$, the set $V \setminus X$ is an independent set of $G$. It follows that $V' \setminus X'$ is an independent set of $G'$, hence $X'$ is a vertex cover of $G'$. Its size satisfies $|X'| = k+ d+d \cdot k \cdot (2d-2) = O(k)$, hence the transformation is linear-parameter. The transformation can clearly be implemented in polynomial time.
For correctness, we shall prove that $\chi(G) \leq d$ if and only if $\od_\Fset(G') \leq d$.

Suppose first that $\chi(G) \leq d$, and consider some proper $d$-coloring of $G$ with color set $[d]$.
We extend this coloring to a $d$-coloring of $G'$ as follows.
First, for each $i \in [d]$, we assign the color $i$ to the vertex $z_i$.
Clearly, no edge that connects two of the vertices $z_1, \ldots, z_d$ is monochromatic.
Next, for each $i \in [d]$ and $v \in X$, consider the vertices of the component $H_{i,v}$. The only vertices of $H_{i,v}$ that already received colors are $z_i$ and $v$. By Lemma~\ref{lemma:gadget}, this partial coloring of $H_{i,v}$ can be extended to a proper $d$-coloring of the whole gadget. Indeed, if $z_i$ and $v$ are assigned the same color then this follows from Item~\ref{itm:same} of the lemma, and if $z_i$ and $v$ are assigned distinct colors then this follows from Item~\ref{itm:distinct} of the lemma.
This gives us a proper $d$-coloring of $G'$, which implies using Claim~\ref{claim:OD<=CHI} that $\od_\Fset(G') \leq \chi(G') \leq d$.

For the converse direction, suppose that $\od_\Fset(G') \leq d$, and consider a $d$-dimensional orthogonal representation $(u_v)_{v \in V'}$ of $G'$ over $\Fset$.
Since the vertices $z_1, \ldots, z_d$ form a clique in $G'$, it follows that their vectors $u_{z_1}, \ldots, u_{z_d}$ are pairwise orthogonal. Since they are non-self-orthogonal, it follows that they are linearly independent, and thus span the entire vector space $\Fset^d$.
For each $i \in [d]$ and $v \in X$, consider the vectors assigned by the given orthogonal representation to the vertices of the component $H_{i,v}$ in $G'$, and apply Item~\ref{itm:od} of Lemma~\ref{lemma:gadget} to obtain that the vectors $u_{z_i}$ and $u_v$ are either orthogonal or proportional.
However, the vectors $u_{z_1}, \ldots, u_{z_d}$ span the vector space $\Fset^d$, hence for each $v \in X$, the nonzero vector $u_v$ cannot be orthogonal to all of them. This yields that for each vertex $v \in X$, the vector $u_v$ is proportional to exactly one of the vectors $u_{z_1}, \ldots, u_{z_d}$.

We define a $d$-coloring of $G$ as follows.
To each vertex $v \in X$, assign the color $i \in [d]$ for which $u_v$ is proportional to $u_{z_i}$.
Since the given orthogonal representation assigns orthogonal vectors to adjacent vertices, it follows that this coloring assigns distinct colors to adjacent vertices in $X$.
Next, to each vertex $v \in V \setminus X$, assign a color from $[d]$ that does not appear on its neighbors.
Notice that all the neighbors of $v$ lie in $X$ and were already colored, because $X$ is a vertex cover of $G$.
To see that such a color exists, recall that the vector $u_v$ is nonzero and orthogonal to the vectors associated with its neighbors in $X$ by the given orthogonal representation of $G'$. Since every such vector is proportional to one of $u_{z_1}, \ldots, u_{z_d}$, it follows that there exists some $i \in [d]$ for which no neighbor of $v$ is associated with a vector proportional to $u_{z_i}$, yielding the existence of the desired color for $v$.
This gives us a proper $d$-coloring of $G$ and implies that $\chi(G) \leq d$, so we are done.
\end{proof}

We finally combine Theorems~\ref{thm:NoKernelCol} and~\ref{thm:reductionColOD} to derive Theorem~\ref{thm:Lower_d}.

\begin{proof}[ of Theorem~\ref{thm:Lower_d}]
Fix a field $\Fset$, an integer $d \geq 3$, and a real $\eps>0$.
By Theorem~\ref{thm:reductionColOD}, there exists a linear-parameter transformation from $d$-$\Col$ on $\Empty+k\mathrm{v}$ graphs to $d\dODP_\Fset$ on $\Empty+k\mathrm{v}$ graphs. Therefore, if $d\dODP_\Fset$ on $\Empty+k\mathrm{v}$ graphs admits a compression with bit-size $O(k^{d-1-\eps})$, then by composing this compression with the given transformation, it follows that $d\qCol$ on $\Empty+k\mathrm{v}$ graphs admits a compression with bit-size $O(k^{d-1-\eps})$ as well.
By Theorem~\ref{thm:NoKernelCol}, this implies that $\NP \subseteq \coNPpoly$, and we are done.
\end{proof}

\subsection{The \texorpdfstring{$3\dODP_\Fset$}{3-Ortho-Dim F} Problem on \texorpdfstring{$\Path+k\mathrm{v}$}{PATH+kv} Graphs}

The proof strategy of Theorem~\ref{thm:Lower_d} enables us to establish the following lower bound on the size of any compression for the $3\dODP_\Fset$ problem on $\Path+k\mathrm{v}$ graphs.

\begin{theorem}\label{thm:ODPath}
For every field $\Fset$, the $3\dODP_\Fset$ problem on $\Path+k\mathrm{v}$ graphs does not admit a polynomial compression unless $\NP \subseteq \coNPpoly$.
\end{theorem}

As an immediate consequence of Theorem~\ref{thm:ODPath}, we obtain the following.

\begin{corollary}
For every field $\Fset$, the $3\dODP_\Fset$ problem on $\Path+k\mathrm{v}$ graphs does not admit a polynomial kernel unless $\NP \subseteq \coNPpoly$.
\end{corollary}

The proof of Theorem~\ref{thm:ODPath} relies on the following theorem, which indicates that it is unlikely that the $3\qCol$ problem on $\Path+k\mathrm{v}$ graphs admits a polynomial compression. This theorem follows from a result of~\cite{DellM14}, saying that the satisfiability problem $\SAT$ parameterized by the number of variables admits no polynomial compression unless $\NP \subseteq \coNPpoly$, combined with a polynomial-parameter transformation of~\cite{JansenK13} to the $3\qCol$ problem on $\Path+k\mathrm{v}$ graphs.

\begin{theorem}\label{thm:ColPath}
The $3\qCol$ problem on $\Path+k\mathrm{v}$ graphs does not admit a polynomial compression unless $\NP \subseteq \coNPpoly$.
\end{theorem}

By another application of Lemma~\ref{lemma:gadget}, we relate the $3$-$\Col$ and $3\dODP_\Fset$ problems on $\Path+k\mathrm{v}$ graphs, as stated below.

\begin{theorem}\label{thm:reductionColODPath}
For every field $\Fset$, there exists a linear-parameter transformation from $3$-$\Col$ on $\Path+k\mathrm{v}$ graphs to $3\dODP_\Fset$ on $\Path+k\mathrm{v}$ graphs.
\end{theorem}

\begin{proof}
Fix a field $\Fset$.
We use here the transformation applied in the proof of Theorem~\ref{thm:reductionColOD} for $d=3$ and analyze it with respect to the $\Path+k\mathrm{v}$ parameterization.
Consider an instance of the $3$-$\Col$ problem on $\Path+k\mathrm{v}$ graphs, namely, a graph $G=(V,E)$ and a set $X \subseteq V$ of size $|X|=k$, such that the graph $G \setminus X$ is a path.
We first add to $G$ a triangle whose vertices are denoted by $z_1, z_2, z_3$.
Then, for each index $i \in [3]$ and each vertex $v \in X$, we add to the graph a copy $H_{i,v}$ of the complement $\overline{C_{6}}$ of the cycle graph on $6$ vertices, where two consecutive vertices of the cycle are identified with the vertices $z_i$ and $v$. Note that we add here $3k$ such gadgets to the graph and that each of them involves $4$ new vertices. Let $G' = (V',E')$ denote the obtained graph, and let $X'$ denote the set that consists of the vertices of $X$ and the vertices that were added to $G$ in the construction. The transformation returns the pair $(G',X')$. Since $G \setminus X$ is a path, so is $G' \setminus X'$. It further holds that $|X'| = k+ 3+3k \cdot 4 = O(k)$, hence the transformation is linear-parameter. The transformation can clearly be implemented in polynomial time.
For correctness, we shall prove that $\chi(G) \leq 3$ if and only if $\od_\Fset(G') \leq 3$.

Suppose first that $\chi(G) \leq 3$, and consider some proper $3$-coloring of $G$ with color set $[3]$.
We extend this coloring to a $3$-coloring of $G'$ as follows.
First, for each $i \in [3]$, we assign the color $i$ to the vertex $z_i$.
Clearly, no edge of the triangle of $z_1, z_2, z_3$ is monochromatic.
Next, for each $i \in [3]$ and $v \in X$, consider the vertices of the component $H_{i,v}$. The only vertices of $H_{i,v}$ that already received colors are $z_i$ and $v$. By Items~\ref{itm:same} and~\ref{itm:distinct} of Lemma~\ref{lemma:gadget}, this partial coloring of $H_{i,v}$ can be extended to a proper $3$-coloring of the whole gadget.
This gives us a proper $3$-coloring of $G'$, which implies using Claim~\ref{claim:OD<=CHI} that $\od_\Fset(G') \leq \chi(G') \leq 3$.

For the converse direction, suppose that $\od_\Fset(G') \leq 3$, and consider a $3$-dimensional orthogonal representation $(u_v)_{v \in V'}$ of $G'$ over $\Fset$.
The vectors $u_{z_1}, u_{z_2}, u_{z_3}$ are non-self-orthogonal and pairwise orthogonal, hence they are linearly independent and thus span the entire vector space $\Fset^3$. For each $i \in [3]$ and $v \in X$, consider the vectors assigned by the given orthogonal representation to the vertices of the component $H_{i,v}$ in $G'$, and apply Item~\ref{itm:od} of Lemma~\ref{lemma:gadget} to obtain that the vectors $u_{z_i}$ and $u_v$ are either orthogonal or proportional. It follows that for each vertex $v \in X$, the vector $u_v$ is proportional to exactly one of the vectors $u_{z_1}, u_{z_2}, u_{z_3}$.

We define a $3$-coloring of $G$ as follows.
To each vertex $v \in X$, assign the color $i \in [3]$ for which $u_v$ is proportional to $u_{z_i}$.
Since the given orthogonal representation assigns orthogonal vectors to adjacent vertices, it follows that this coloring assigns distinct colors to adjacent vertices in $X$.
Recalling that $G \setminus X$ is a path, let $v_1, \ldots, v_\ell$ denote its vertices according to their order along the path.
We extend the $3$-coloring of $X$ to the vertices of the path, one by one, as follows.
For some $j \in [\ell]$, suppose that we have already assigned colors to $v_1, \ldots, v_{j-1}$.
The vector $u_{v_j}$ can be written as $u_{v_j} = \alpha_1 \cdot u_{z_1} + \alpha_2 \cdot u_{z_2} + \alpha_3 \cdot u_{z_3}$ for some coefficients $\alpha_1, \alpha_2, \alpha_3 \in \Fset$, because $u_{z_1},u_{z_2},u_{z_3}$ span $\Fset^3$.
We assign to $v_j$ a color $i \in [3]$ with $\alpha_i \neq 0$ that differs from the color of the preceding vertex $v_{j-1}$ in the path (where the latter condition is ignored for $j=1$).
Observe that the condition $\alpha_i \neq 0$ guarantees that the color assigned to $v_j$ is distinct from the colors of its neighbors in $X$, because their vectors are orthogonal to $u_{v_j}$ and because each of them is proportional to one of $u_{z_1}, u_{z_2}, u_{z_3}$.
The condition that $i$ differs from the color of $v_{j-1}$ guarantees that the color of $v_j$ is different from that of its already colored neighbor in the path.
It remains to show that there exists a color that satisfies the two conditions.
First, since $u_{v_j}$ is nonzero, there exists some index $i \in [3]$ satisfying $\alpha_i \neq 0$.
If there are at least two such indices $i$, then it is clearly possible to choose one that differs from the color of $v_{j-1}$. Otherwise, $u_{v_j} = \alpha_i \cdot u_{z_i}$ for some $i \in [3]$ with $\alpha_i \neq 0$. Writing the vector $u_{v_{j-1}}$ as a linear combination of $u_{z_1}, u_{z_2}, u_{z_3}$, the coefficient of $u_{z_i}$ must be zero, because $u_{v_j}$ is orthogonal to $u_{v_{j-1}}$. By the definition of our coloring, this implies that the vertex $v_{j-1}$ is not assigned color $i$, so the color $i$ can be assigned to $v_j$.
This completes the proof.
\end{proof}

We finally combine Theorems~\ref{thm:ColPath} and~\ref{thm:reductionColODPath} to derive Theorem~\ref{thm:ODPath}.

\begin{proof}[ of Theorem~\ref{thm:ODPath}]

Fix a field $\Fset$.
By Theorem~\ref{thm:reductionColODPath}, there exists a linear-parameter transformation from $3$-$\Col$ on $\Path+k\mathrm{v}$ graphs to $3\dODP_\Fset$ on $\Path+k\mathrm{v}$ graphs.
Therefore, if $3\dODP_\Fset$ on $\Path+k\mathrm{v}$ graphs admits a polynomial compression, then by composing this compression with the given transformation, it follows that $3\qCol$ on $\Path+k\mathrm{v}$ graphs admits a polynomial compression as well.
By Theorem~\ref{thm:ColPath}, this implies that $\NP \subseteq \coNPpoly$, and we are done.
\end{proof}

\section{Kernelization for \texorpdfstring{$d\dODP_\Fset$}{d-Ortho-Dim F} with Structural Parameterizations}\label{sec:KernelG}

We study now the kernelizability of the $d\dODP_\Fset$ problem on $\calG+k\mathrm{v}$ graphs for a general hereditary family of graphs $\calG$, extending the case of $\calG = \Empty$ considered in the previous sections.
To do so, we need the following variant of the $d\dODP_\Fset$ problem.

\begin{definition}\label{def:SC}
For a field $\Fset$ and an integer $d$, define the $d\dSubChoose_\Fset$ problem as follows.
An instance of the problem is a pair $(G,L)$ of a graph $G=(V,E)$ and a function $L$ that maps each vertex $v \in V$ to a subspace $L(v)$ of $\Fset^d$.
The goal is to decide whether there exists a $d$-dimensional orthogonal representation $(u_v)_{v \in V}$ of $G$ over $\Fset$ satisfying $u_v \in L(v)$ for all $v \in V$.
\end{definition}
\noindent
For instances $(G,L)$ and $(G',L')$ of $d\dSubChoose_\Fset$, we say that $(G',L')$ is a sub-instance of $(G,L)$ if $G'$ is an induced subgraph of $G$ and for every vertex $v$ of $G'$ it holds that $L(v)=L'(v)$.
For a function $g: \N \rightarrow \N$, we say that a family of graphs $\calG$ has $g$-size $\NO$-certificates for $\SubChoose_\Fset$ if for every integer $d$, every $\NO$ instance $(G,L)$ of $d\dSubChoose_\Fset$ with $G \in \calG$ has a $\NO$ sub-instance $(G',L')$, where $G'$ has at most $g(d)$ vertices.
For example, it can be seen that the family $\Empty$ has $g$-size $\NO$-certificates for $\SubChoose_\Fset$ for the function $g$ defined by $g(d)=1$ for all integers $d$. This indeed follows from the fact that a $\NO$ instance $(G,L)$ of the problem with $G \in \Empty$ must have a vertex $v$, such that all the vectors of $L(v)$ are self-orthogonal.

The following theorem shows that if a hereditary family of graphs $\calG$ has $g$-size $\NO$-certificates for $\SubChoose_\Fset$ for some function $g$, then the $d\dODP_\Fset$ problem on $\calG+k\mathrm{v}$ graphs admits a polynomial kernel for every integer $d$.
Its proof employs the technique used in~\cite{JansenK13} for the analogous coloring problem.
Note that the special case of the theorem corresponding to $\calG = \Empty$ coincides with Theorem~\ref{thm:IntroUpperF}.
\begin{theorem}\label{thm:KernelGenG}
For a field $\Fset$, a function $g: \N \rightarrow \N$, and a hereditary family of graphs $\calG$, suppose that $\calG$ has $g$-size $\NO$-certificates for $\SubChoose_\Fset$.
Then for every integer $d$, the $d\dODP_\Fset$ problem on $\calG+k\mathrm{v}$ graphs admits a kernel with $O(k^{d \cdot g(d)})$ vertices.
\end{theorem}

\begin{proof}
Fix an integer $d$, and consider an instance of the $d\dODP_\Fset$ problem on $\calG+k\mathrm{v}$ graphs, i.e., a graph $G= (V,E)$ and a set $X \subseteq V$ of size $|X|=k$ such that $G \setminus X \in \calG$. Consider the algorithm that given such an instance acts as follows.
\begin{enumerate}
  \item\label{itm:1kernelG} For every graph $F$ on $t \leq g(d)$ vertices $h_1, \ldots, h_t$ and for every $t$-tuple $(S_1,\ldots,S_t)$ of sets of at most $d$ vertices of $X$, if there is an induced subgraph of $G \setminus X$ on vertices $v_1, \ldots,v_t$ that is isomorphic to $F$ by the mapping $v_i \mapsto h_i$ and with $S_i \subseteq N_G(v_i)$ for all $i \in [t]$, then mark the vertices $v_1, \ldots, v_t$ for one such subgraph chosen arbitrarily.
  \item Let $Y$ denote the set of all marked vertices of $G \setminus X$, and set $V' = X \cup Y$.
  \item Return the instance $(G',X)$ where $G' = G[V']$.
\end{enumerate}
\noindent
Note that $G'$ is an induced subgraph of $G$. By $G \setminus X \in \calG$, using the fact that $\calG$ is hereditary, it follows that $G' \setminus X \in \calG$, hence $(G',X)$ is an appropriate instance of the problem.

We first observe that the produced graph $G'$ has $O(k^{d \cdot g(d)})$ vertices and that the algorithm can be implemented in polynomial time.
Indeed, every iteration of Step~\ref{itm:1kernelG} of the algorithm is associated with a graph $F$ on $t \leq g(d)$ vertices and with a $t$-tuple of sets of at most $d$ vertices of $X$. Recalling that $d$ is a fixed constant and that $|X|=k$, it follows that the number of iterations is bounded by $O(k^{d \cdot g(d)})$. In every such iteration, at most $g(d)$ vertices are marked, hence the total number of vertices in $G'$ is $|X|+|Y| = k+O(k^{d \cdot g(d)}) = O(k^{d \cdot g(d)})$. As for the running time, an iteration associated with a graph $F$ on $t$ vertices and a $t$-tuple $(S_1, \ldots, S_t)$ can be performed by enumerating  all the $t$-subsets of $V \setminus X$ and checking whether the corresponding induced subgraph of $G$ is isomorphic to $F$. Since $d$ is a fixed constant, using $t \leq g(d)$, this can be done in polynomial time by trying all the permutations of the vertex set of $F$. As explained above, the number of iterations is also polynomial, hence the algorithm can be implemented in polynomial time.

We next prove the correctness of the kernel, namely, $\od_\Fset(G) \leq d$ if and only if $\od_\Fset(G') \leq d$.
It clearly holds that if  $\od_\Fset(G) \leq d$ then $\od_\Fset(G') \leq d$, because $G'$ is a subgraph of $G$.
For the converse, suppose that $\od_\Fset(G') \leq d$, that is, there exists a $d$-dimensional orthogonal representation $(u_v)_{v \in V'}$ of $G'$ over $\Fset$.
To define a $d$-dimensional orthogonal representation of $G$ over $\Fset$, we first assign to each vertex $v \in X$ the vector $u_v$. Every two adjacent vertices in $G[X]$ are obviously assigned orthogonal vectors, and we would like to extend this partial assignment to a $d$-dimensional orthogonal representation of $G$ over $\Fset$.
This task may be represented in terms of the $d\dSubChoose_\Fset$ problem.
Indeed, let $H = G \setminus X$ be the subgraph of $G$ induced by $V_H = V \setminus X$, and assign to every vertex $v \in V_H$ the subspace $L(v) \subseteq \Fset^d$ that consists of all the vectors that are orthogonal to the vectors assigned to the neighbors of $v$ in $X$, that is, \[L(v) = \linspan(\{ u_{v'} \mid v' \in X \cap N_G(v) \})^\perp.\]
It follows that our partial orthogonal representation of $G$ can be extended to the whole graph if and only if $(H,L)$ is a $\YES$ instance of $d\dSubChoose_\Fset$.
Suppose for contradiction that $(H,L)$ is a $\NO$ instance of $d\dSubChoose_\Fset$.
Since $H \in \calG$, our assumption on $\calG$ implies that $(H,L)$ has a $\NO$ sub-instance $(H',L')$, where $H'$ has $t$ vertices for some $t \leq g(d)$.

Denote the vertices of $H'$ by $h_1, \ldots, h_t$.
For an $i \in [t]$, let $W_i = \linspan(\{ u_{v'} \mid v' \in X \cap N_G(h_i) \})$ denote the subspace spanned by the vectors assigned to the vertices of $X$ that are adjacent to $h_i$ in $G$. It follows that there exists a set $S_i \subseteq X \cap N_G(h_i)$ such that the vectors $u_{v'}$ with $v' \in S_i$ form a basis of the subspace $W_i$ of $\Fset^d$, and it clearly holds that $|S_i| \leq d$.
Since $H'$ is an induced subgraph of $G \setminus X$ with at most $g(d)$ vertices, and since the $t$-tuple $(S_1, \ldots, S_t)$ satisfies $S_i \subseteq X \cap N_G(h_i)$ and $|S_i| \leq d$ for all $i \in [t]$, one of the iterations of our algorithm must have marked the vertices of an induced subgraph $H^*$ of $G \setminus X$ on vertices $v_1, \ldots, v_t$, such that
$H^*$ is isomorphic to $H'$ by the mapping $v_i \mapsto h_i$ and, in addition, $S_i \subseteq X \cap N_G(v_i)$ for all $i \in [t]$.
It follows that the vertices of $H^*$ are included in the graph $G'$ and that for each $i \in [t]$, the vectors assigned by the given orthogonal representation of $G'$ to the neighbors of $v_i$ in $X$ span $W_i$.
Therefore, the restriction of the given orthogonal representation of $G'$ to the vertices of $H^*$ assigns to each vertex $v_i$ a vector from $W_i^\perp$.
This gives us an orthogonal representation of $H'$, such that for each $i \in [t]$, the vertex $h_i$ is assigned a vector from $L(h_i)$, and thus from $L'(h_i)$, contradicting the fact that $(H',L')$ is a $\NO$ instance of $d\dSubChoose_\Fset$. This completes the proof.
\end{proof}

\subsection{\texorpdfstring{$\NO$}{NO}-certificates for \texorpdfstring{$d\dSubChoose_\Fset$}{d-Subspace Choosability F}}

In what follows, we demonstrate the applicability of Theorem~\ref{thm:KernelGenG}.
To do so, we consider a couple of hereditary families of graphs and prove upper bounds on the size of their $\NO$-certificates for the $d\dSubChoose_\Fset$ problem.

\subsubsection{Split Graphs}

Consider the following definition.

\begin{definition}\label{def:split}
A graph $G=(V,E)$ is called a split graph if there exists a partition $V = C \cup I$ of its vertex set into a clique $C$ and an independent set $I$.
Let $\Split$ denote the family of all split graphs, and let $\USplit$ denote the family of all graphs whose connected components are members of $\Split$.
\end{definition}
\noindent
Note that the family $\USplit$ is hereditary.

For integers $d$, $k$, and $q$, the $q$-binomial coefficient ${\rectbinom{d}{k}}_q$ is defined by $\rectbinom{d}{k}_q = \prod_{i=0}^{k-1}{\frac{q^{d-i}-1}{q^{k-i}-1}}$.
For a finite field $\Fset$ of order $q$, ${\rectbinom{d}{k}}_q$ is the number of $k$-dimensional subspaces of $\Fset^d$.
We will use the following Sperner-type result on subspaces over finite fields (see, e.g.,~\cite[Example~4.6.2]{SpernerBook97}).

\begin{theorem}\label{thm:SpernerSub}
Let $\Fset$ be a finite field of order $q$, and let $d$ be an integer.
Let $\calW$ be a collection of subspaces of $\Fset^d$, such that no subspace in $\calW$ is contained in another subspace in $\calW$.
Then $|\calW| \leq \rectbinom{d}{\lfloor d/2 \rfloor }_{q}$.
\end{theorem}

For the $\SubChoose_\Fset$ problem over finite fields $\Fset$, we prove the following result.
\begin{lemma}\label{lemma:NosplitFinite}
Let $\Fset$ be a finite field of order $q$, and let $g:\N \rightarrow \N$ be the function defined by \[g(d) = d+2^d \cdot \rectbinom{d}{\lfloor d/2 \rfloor }_{q}.\]
Then the family $\USplit$ has $g$-size $\NO$-certificates for $\SubChoose_\Fset$.
\end{lemma}

\begin{proof}
Let $d$ be an integer, and let $(G,L)$ be a $\NO$ instance  of $d\dSubChoose_\Fset$ for a graph $G=(V,E)$ that lies in $\USplit$.
Our goal is to show that $(G,L)$ has a $\NO$ sub-instance with at most $g(d)$ vertices.
It may be assumed that $G \in \Split$, because if $G$ is a $\NO$ instance, then one of its connected components is a $\NO$ instance as well.
It follows that there exists a partition $V = C \cup I$ of the vertex set of $G$, where $C$ is a clique and $I$ is an independent set.

To obtain the desired sub-instance of $(G,L)$ we act as follows.
If $|C| \geq d+1$, then any sub-instance of $(G,L)$ induced by $d+1$ vertices of $C$ is clearly a $\NO$ instance, and we are done.
Otherwise, we start with $(G,L)$ and as long as there exist two distinct vertices $v_1,v_2 \in I$ satisfying $N_G(v_1) = N_G(v_2)$ and $L(v_1) \subseteq L(v_2)$, we remove $v_2$ from the graph. Observe that after the removal of $v_2$, the instance remains a $\NO$ instance, because any assignment of a vector to the vertex $v_1$ suits the vertex $v_2$ as well. Let $(G',L')$ denote the $\NO$ sub-instance of $(G,L)$ obtained at the end of the process, and let $V'$ denote the vertex set of $G'$.
To complete the argument, we show that $|V'| \leq g(d)$.
Indeed, for every set $S \subseteq C$, consider the collection of subspaces $\calW_S = \{L'(v) \mid v \in V' \cap I,~N_G(v) = S\}$ of $\Fset^d$.
By construction, no subspace in $\calW_S$ is contained in another subspace in $\calW_S$, hence by Theorem~\ref{thm:SpernerSub}, it holds that $|\calW_S| \leq \rectbinom{d}{\lfloor d/2 \rfloor }_{q}$. Using $|C| \leq d$, we obtain that
\[|V'| \leq d+ \sum_{S \subseteq C}{|\calW_S|} \leq d+ 2^d \cdot \rectbinom{d}{\lfloor d/2 \rfloor }_{q} = g(d),\] so we are done.
\end{proof}

We turn to the more nuanced case of the $\SubChoose_\Fset$ problem over infinite fields $\Fset$.
We assume here that the field $\Fset$ satisfies that no nonzero vector over $\Fset$ is self-orthogonal (as is the case for the real field $\R$).
We need the following theorem of Kalai~\cite{Kalai80}.
\begin{theorem}[\cite{Kalai80}]\label{thm:Kalai}
For a field $\Fset$ and an integer $d$, let $(P_1,Q_1), \ldots, (P_m,Q_m)$ be $m$ pairs of subspaces of $\Fset^d$ such that
\begin{enumerate}
  \item $P_i \cap Q_i = \{0\}$ for every $i \in [m]$, and
  \item $P_i \cap Q_j \neq \{0\}$ for every $i \neq j \in [m]$.
\end{enumerate}
Then $m \leq \binom{d}{\lfloor d/2 \rfloor }$.
\end{theorem}

We also need the following claim.

\begin{claim}\label{claim:clean}
Let $\Fset$ be a field and let $d$ be an integer, such that no nonzero vector of $\Fset^d$ is self-orthogonal.
Let $G=(V,E)$ be a graph, and let $(G,L)$ be an instance of $d\dSubChoose_\Fset$.
Let $w \in V$ be a vertex of $G$, such that for every subspace $Q \subseteq \Fset^d$ that satisfies $Q \cap L(v) \neq \{0\}$ for all vertices $v \in V \setminus \{w\}$ with $N_G(v) = N_G(w)$, it holds that $Q \cap L(w) \neq \{0\}$.
Let $(G',L')$ be the sub-instance of $(G,L)$ obtained by removing the vertex $w$.
Then, $(G,L)$ is a $\YES$ instance if and only if $(G',L')$ is a $\YES$ instance.
\end{claim}

\begin{proof}
It obviously holds that if $(G,L)$ is a $\YES$ instance, then so is its sub-instance $(G',L')$.
For the other direction, suppose that $(G',L')$ is a $\YES$ instance.
Therefore, there exists a $d$-dimensional orthogonal representation of $G'$ over $\Fset$ that assigns a nonzero vector $u_v \in L(v)$ to each vertex $v \in V \setminus \{w\}$.
To extend this assignment to an orthogonal representation of $G$, one has to show that there exists a non-self-orthogonal vector $u_w \in L(w)$ that is orthogonal to the vectors associated with the vertices of $N_G(w)$, equivalently, that lies in the subspace $Q = \linspan(\{u_v ~|~ v \in N_G(w)\})^\perp$.
For each vertex $v \in V \setminus \{w\}$ with $N_G(v)=N_G(w)$, the vector $u_v$ is nonzero and lies in $Q$, hence $Q \cap L(v) \neq \{0\}$.
By the assumption of the claim, it follows that $Q \cap L(w) \neq \{0\}$.
Recalling that $\Fset^d$ has no nonzero self-orthogonal vectors, this yields the existence of the desired vector for the vertex $w$ and implies that $(G,L)$ is a $\YES$ instance.
\end{proof}

\begin{lemma}\label{lemma:NosplitInfinite}
Let $\Fset$ be a field such that for every integer $d$, no nonzero vector of $\Fset^d$ is self-orthogonal.
Let $g:\N \rightarrow \N$ be the function defined by $g(d) = d+2^d \cdot \binom{d}{\lfloor d/2 \rfloor}$.
Then the family $\USplit$ has $g$-size $\NO$-certificates for $\SubChoose_\Fset$.
\end{lemma}
\begin{proof}
Consider a $\NO$ instance $(G,L)$ of $d\dSubChoose_\Fset$ for a graph $G=(V,E)$ that lies in $\USplit$.
As before, it may be assumed that $G \in \Split$, hence there exists a partition $V = C \cup I$ of the vertex set of $G$, where $C$ is a clique and $I$ is an independent set.
To obtain a $\NO$ sub-instance of $(G,L)$ with at most $g(d)$ vertices, we act as follows.
If $|C| \geq d+1$, then any sub-instance of $(G,L)$ induced by $d+1$ vertices of $C$ is clearly a $\NO$ instance, and we are done.
Otherwise, we start with $(G,L)$ and as long as there exists a vertex $w \in I$ such that for every subspace $Q \subseteq \Fset^d$ that satisfies $Q \cap L(v) \neq \{0\}$ for all vertices $v \in V \setminus \{w\}$ with $N_G(v)=N_G(w)$, it holds that $Q \cap L(w) \neq \{0\}$, we remove the vertex $w$ from the current instance (as well as from $V$ and $I$). By Claim~\ref{claim:clean}, the instance remains a $\NO$ instance after the removal of every such vertex $w$.
Let $(G',L')$ denote the sub-instance of $(G,L)$ obtained at the end of this process, and let $I'$ denote the set of the remaining vertices of $I$ in $G'$.

We prove that $G'$ has at most $g(d)$ vertices.
Recall that $|C| \leq d$.
For any set $S \subseteq C$, we define $I_S = \{ w \in I' ~|~ N_{G'}(w) = S\}$ as the set of vertices of $I'$ whose neighbors in $G'$ are the vertices of $S$.
With each vertex $w \in I_S$, we associate a pair $(P_w,Q_w)$ of subspaces of $\Fset^d$ as follows.
The subspace $P_w$ is simply defined to be $P_w = L(w)$.
The subspace $Q_w$ is defined to be some subspace of $\Fset^d$ satisfying $Q_w \cap L(v) \neq \{0\}$ for all vertices $v \in I_S \setminus \{w\}$ as well as $Q_w \cap L(w) = \{0\}$.
Notice that such a subspace $Q_w$ is guaranteed to exist, because otherwise $w$ would be removed from the instance $(G',L')$ in the above process.
Now, the pairs $(P_w,Q_w)$ with $w \in I_S$ satisfy $P_w \cap Q_w = \{0\}$ for every $w \in I_S$ and $P_w \cap Q_{w'} \neq \{0\}$ for every distinct $w, w' \in I_S$. By Theorem~\ref{thm:Kalai}, this implies that $|I_S| \leq \binom{d}{\lfloor d/2 \rfloor}$.
Since the number of sets $S \subseteq C$ does not exceed $2^d$, it follows that $|I'| = \sum_{S \subseteq C}{|I_S|} \leq 2^d \cdot \binom{d}{\lfloor d/2 \rfloor}$.
Therefore, together with the vertices of $C$, the number of vertices of $G'$ does not exceed $g(d)$, and we are done.
\end{proof}
\noindent
Combining Theorem~\ref{thm:KernelGenG} with Lemmas~\ref{lemma:NosplitFinite} and~\ref{lemma:NosplitInfinite} yields, for every integer $d$ and for various fields $\Fset$, a polynomial kernel for the $d\dODP_\Fset$ problem on $\USplit+k\mathrm{v}$ graphs.

We conclude this section with a lower bound, for any field $\Fset$, on the size of $\NO$-certificates for $\SubChoose_\Fset$ on instances that involve split graphs.
Here, a $\NO$ instance is called irreducible if any removal of a vertex turns it into a $\YES$ instance.

\begin{lemma}
For every field $\Fset$ and for every integer $d$, there exists an irreducible $\NO$ instance $(G,L)$ of $d\dSubChoose_\Fset$, where the graph $G$ lies in $\Split$ and has $d + \binom{d}{\lfloor d/2 \rfloor}$ vertices.
\end{lemma}
\begin{proof}
Let $G$ denote the graph that consists of a clique $C$ of $d$ vertices denoted by $c_1, \ldots, c_d$ and an independent set $I$ whose vertices are all the subsets of $[d]$ of size $\lfloor d/2 \rfloor$, where every vertex $S \in I$ is adjacent to the vertices $c_i$ with $i \in S$. Note that $G$ is a split graph and that $|I| = \binom{d}{\lfloor d/2 \rfloor}$, hence the number of vertices in $G$ is $d + \binom{d}{\lfloor d/2 \rfloor}$.
Let $W$ denote the subspace of $\Fset^d$ spanned by the vectors $e_i$ with $1 \leq i \leq \lfloor d/2 \rfloor$, where $e_i$ stands for the vector of $\Fset^d$ with $1$ on the $i$th entry and $0$ everywhere else.
Let $L$ denote the function that maps every vertex of $C$ to the entire vector space $\Fset^d$ and every vertex of $I$ to $W$.
We shall prove that $(G,L)$ is an irreducible $\NO$ instance of $d\dSubChoose_\Fset$.

We first prove that $(G,L)$ is a $\NO$ instance of $d\dSubChoose_\Fset$.
Suppose for contradiction that there exists a $d$-dimensional orthogonal representation of $G$ over $\Fset$ that assigns a vector $u_v \in L(v)$ to each vertex $v \in C \cup I$.
The vectors $u_{c_1}, \ldots, u_{c_d}$ assigned to the vertices of the clique $C$ are non-self-orthogonal and pairwise orthogonal, hence they are linearly independent and thus span the entire vector space $\Fset^d$. It thus follows that the vectors $e_i$ with $\lfloor d/2 \rfloor +1 \leq i \leq d$ can be extended to a basis of $\Fset^d$ by the vectors $u_{c_i}$ with $i \in S$ for some set $S \subseteq [d]$ of size $|S| = \lfloor d/2 \rfloor$. However, the vector $u_S$ assigned by the given orthogonal representation to the vertex $S$ of $I$ must be orthogonal to the vectors $u_{c_i}$ with $i \in S$ of its neighbors as well as to the vectors $e_i$ with $\lfloor d/2 \rfloor +1 \leq i \leq d$ (because $u_S \in W$). It follows that $u_S$ must be orthogonal to every vector of $\Fset^d$, in contradiction to the fact that it is non-self-orthogonal.

We next show that any removal of a vertex $v \in C \cup I$ from $(G,L)$ results in a $\YES$ instance of $d\dSubChoose_\Fset$.
First, if $v \in C$, assign to the remaining vertices of $C$ the vectors $e_2, \ldots, e_d$, and to every vertex of $I$ the vector $e_1$.
This assignment is clearly an orthogonal representation of $G$ that associates each vertex with a vector from its subspace.
Otherwise, if $v \in I$, then let $S \subseteq [d]$ denote the set associated with $v$, assign the vectors $e_1, \ldots, e_{\lfloor d/2 \rfloor}$ to the vertices $c_i$ of $C$ with $i \in S$ and the vectors $e_{\lfloor d/2 \rfloor+1}, \ldots, e_d$ to the remaining vertices of $C$.
Further, for every vertex $S' \in I$ with $S' \neq S$, there exists some $i \in S \setminus S'$, so we assign to $S'$ the vector $u_{c_i}$. By $i \in S \setminus S'$, the vector $u_{c_i}$ lies in $W$ and is orthogonal to the vectors of the neighbors of $S'$ in $C$.
This gives us the desired orthogonal representation and completes the proof.
\end{proof}

\subsubsection{Cochordal Graphs}

We next study the size of $\NO$-certificates for the $\SubChoose_\Fset$ problem on instances that involve cochordal graphs, defined as follows.

\begin{definition}\label{def:cochordal}
A chordal graph is a graph with no induced cycle of length at least $4$, equivalently, every cycle of length at least $4$ in the graph has a chord.
A cochordal graph is the complement of a chordal graph.
Let $\Cochordal$ denote the family of all cochordal graphs, and let $\UCochordal$ denote the family of all graphs whose connected components are members of $\Cochordal$.
\end{definition}
\noindent
Note that the family $\UCochordal$ is hereditary and contains the family $\USplit$.

A fundamental property of chordal graphs is given by the following statement.

\begin{proposition}[{\cite[Chapter~1.2]{GraphClassesBook}}]\label{prop:simplicial}
Every chordal graph has a vertex whose neighbors form a clique.
\end{proposition}

We need the following claim.
\begin{claim}\label{claim:UCochordal}
For a field $\Fset$ and an integer $d$, let $(G,L)$ be an instance of $d\dSubChoose_\Fset$, let $v$ be a vertex of $G$ whose non-neighbors form an independent set, and let $u \in L(v)$ be a non-self-orthogonal vector.
Let $G'$ be the graph obtained from $G$ by removing the vertex $v$ and all of its non-neighbors $v'$ with $u \in L(v')$.
Let $L'$ be the function that maps any vertex $v'$ of $G'$ to $L(v') \cap u^{\perp}$ if $v'$ is adjacent to $v$ in $G$ and to $L(v')$ otherwise.
Then, $(G,L)$ has a solution that assigns the vector $u$ to the vertex $v$ if and only if $(G',L')$ is a $\YES$ instance of $d\dSubChoose_\Fset$.
\end{claim}

\begin{proof}
Suppose first that $(G,L)$ has a solution that assigns the vector $u$ to the vertex $v$.
We claim that the restriction of this solution to the vertices of $G'$ is a solution for $(G',L')$.
Indeed, $G'$ is a subgraph of $G$, hence an orthogonal representation of $G$ induces an orthogonal representation of $G'$.
We further observe that for every vertex $v'$ of $G'$, the vector assigned to $v'$ lies in the subspace $L'(v')$.
If $v'$ is adjacent to $v$ in $G$, then its vector is orthogonal to the vector $u$ of $v$ and thus lies in $L'(v') = L(v') \cap u^{\perp}$, and otherwise, it certainly lies in $L'(v') = L(v')$.
Therefore, $(G',L')$ is a $\YES$ instance of $d\dSubChoose_\Fset$.

Suppose next that $(G',L')$ is a $\YES$ instance of $d\dSubChoose_\Fset$, and consider an orthogonal representation of $G'$ that assigns to every vertex $v'$ a vector from $L'(v') \subseteq L(v')$. We extend this orthogonal representation to $G$ by assigning the vector $u$ to the vertices of $G$ that are not in $G'$, that is, the vertex $v$ and its non-neighbors $v'$ in $G$ with $u \in L(v')$. By construction, the vector $u$ lies in the subspaces associated by $L$ with these vertices.
It remains to show that the vectors of their neighbors are all orthogonal to $u$.
Firstly, by the definition of $L'$, the vectors of the neighbors of $v$ in $G$ are all orthogonal to $u$.
Secondly, since the non-neighbors of $v$ in $G$ form an independent set, all of their neighbors are adjacent to $v$, hence their vectors are orthogonal to $u$ as well.
This implies that the obtained assignment forms a solution for $(G,L)$ that assigns the vector $u$ to the vertex $v$, so we are done.
\end{proof}

\begin{remark}\label{remark:G'L'}
The instance $(G',L')$ defined in Claim~\ref{claim:UCochordal} satisfies $u \notin L'(v')$ for all vertices $v'$ of $G'$.
\end{remark}

Equipped with Claim~\ref{claim:UCochordal}, we prove the following lemma.

\begin{lemma}\label{lemma:NocochordalFinite}
Let $\Fset$ be a finite field of order $q$, and let $g:\N \rightarrow \N$ be the function defined by $g(d) = (q^d)!$.
Then the family $\UCochordal$ has $g$-size $\NO$-certificates for $\SubChoose_\Fset$.
\end{lemma}

\begin{proof}
For an integer $d$, let $(G,L)$ be a $\NO$ instance of $d\dSubChoose_\Fset$ for a graph $G \in \UCochordal$.
Our goal is to show that $(G,L)$ admits a $\NO$ sub-instance with at most $g(d)$ vertices.
To do so, we present an algorithm that given such an instance $(G,L)$ marks vertices in $G$.
We then show that the marked vertices induce a $\NO$ sub-instance and that their number does not exceed $g(d)$.

The algorithm acts as follows. If there exists a vertex $v$ in $G$ such that all the vectors of $L(v)$ are self-orthogonal, then the algorithm marks it and terminates.
Otherwise, since $(G,L)$ is a $\NO$ instance, it must have a connected component that induces a $\NO$ sub-instance, allowing us to assume without loss of generality that $G$ is connected and is thus cochordal. By Proposition~\ref{prop:simplicial}, there exists a vertex $v$ in $G$ whose non-neighbors form an independent set. The algorithm chooses such a vertex $v$ and marks it. Then, for each non-self-orthogonal vector $u \in L(v)$, the algorithm constructs an instance $(G',L')$, where $G'$ is the graph obtained from $G$ by removing the vertex $v$ and all of its non-neighbors $v'$ with $u \in L(v')$, and $L'$ is the function that maps any vertex $v'$ of $G'$ to $L(v') \cap u^{\perp}$ if $v'$ is adjacent to $v$ in $G$ and to $L(v')$ otherwise.
Then, the algorithm calls itself recursively on the obtained instances.

We argue that if $(G,L)$ is a $\NO$ instance, then the vertices marked by the algorithm while running on $(G,L)$ induce a $\NO$ sub-instance.
We prove it by induction on the number of vertices.
If there exists a vertex $v$ in $G$ such that all the vectors of $L(v)$ are self-orthogonal, then the sub-instance induced by this single vertex is clearly a $\NO$ sub-instance. Since the algorithm marks such a vertex, the marked vertices induce a $\NO$ sub-instance.
This in particular covers the case where the graph $G$ has a single vertex.
Otherwise, assuming that $G$ is connected and using Proposition~\ref{prop:simplicial}, the algorithm chooses a vertex $v$ in $G$ whose non-neighbors form an independent set.
Since $(G,L)$ is a $\NO$ instance, for every non-self-orthogonal vector $u \in L(v)$ there is no solution for $(G,L)$ that assigns the vector $u$ to the vertex $v$.
By Claim~\ref{claim:UCochordal}, this condition is equivalent to whether the smaller instance $(G',L')$ associated with $u$ is a $\NO$ instance.
Since the algorithm iterates over all non-self-orthogonal vectors $u \in L(v)$ and for each of them constructs a suitable instance $(G',L')$ and, by the induction hypothesis, marks vertices that induce a $\NO$ sub-instance of $(G',L')$, the marked vertices form a $\NO$ sub-instance, as required.

We finally show that the number of vertices marked by the algorithm throughout the run on a $\NO$ instance $(G,L)$ does not exceed $g(d)$.
For an integer $t$, let $h(t)$ denote the largest possible number of marked vertices for a $\NO$ instance $(G,L)$ with at most $t$ distinct vectors in all the subspaces of the image of $L$. It clearly holds that $h(1)=1$, because if all the subspaces contain at most one vector, then they are all equal to the zero subspace, and the algorithm marks a single vertex.
If there are at most $t$ vectors in all the subspaces of the image of $L$, then when the algorithm marks a vertex $v$, it iterates over all non-self-orthogonal vectors $u \in L(v)$, whose number is at most $t-1$ (excluding the zero vector), and for each of them calls itself recursively with an instance $(G',L')$ satisfying $u \notin L'(v')$ for all vertices $v'$ of $G'$ (see Remark~\ref{remark:G'L'}). This implies that for every integer $t \geq 2$, it holds that $h(t) \leq 1+ (t-1) \cdot h(t-1)$, and it can be easily checked that $h(t) \leq t!$ for all $t$. Since every instance of $d\dSubChoose_\Fset$ has at most $q^d$ distinct vectors of $\Fset^d$ in its subspaces, it follows that every $\NO$ instance has a $\NO$ sub-instance on at most $h(q^d) \leq (q^d)!$ vertices, and we are done.
\end{proof}
\noindent
Combining Theorem~\ref{thm:KernelGenG} with Lemma~\ref{lemma:NocochordalFinite} yields a polynomial kernel for the $d\dODP_\Fset$ problem on $\UCochordal+k\mathrm{v}$ graphs for every integer $d$ and every finite field $\Fset$.

\section*{Acknowledgments}
We are grateful to the anonymous reviewers for their constructive and helpful feedback.

\bibliographystyle{abbrv}
\bibliography{od_vc}

\end{document}